\documentclass[a4paper,12pt]{article}
\usepackage[latin1]{inputenc}
\usepackage[T1]{fontenc}
\usepackage{amsmath,amssymb}
\usepackage{array,booktabs}
\numberwithin{equation}{section}
\usepackage{graphicx}
\usepackage{fullpage}
\usepackage{tikz}
\usepackage{amsthm}
\usepackage{mdframed}
\usepackage{bbm}
\usepackage{hyperref}
\usepackage{pdfpages}
\usepackage{wasysym}
\def\E{\ensuremath\mathbb{E}}
\def\P{\ensuremath\mathbb{P}}
\def\Z{\ensuremath\mathbb{Z}}
\usepackage{multirow}
\theoremstyle{plain}
\newtheorem{theorem}{Theorem}[]

\newtheorem{lemma}[theorem]{Lemma}
\newtheorem{corollary}[theorem]{Corollary}

\theoremstyle{definition}
\newtheorem{definition}[theorem]{Definition}

\title{Power of $d$ Choices with Simple Tabulation\footnote{This research is supported by Mikkel Thorup's Advanced
Grant DFF-0602-02499B from the Danish Council for Independent
Research and by his Villum Investor grant 16582.}}
\author{Anders Aamand\footnote{BARC, University of Copenhagen, Universitetsparken 1, Copenhagen, Denmark.}, Mathias B\ae k Tejs Knudsen\footnote{University of Copenhagen and Supwiz, Copenhagen, Denmark.}, and Mikkel Thorup$^\dagger$}
\date{}
\begin{document}
\maketitle
\begin{abstract}
Suppose that we are to place $m$ balls into $n$ bins sequentially using the $d$-choice paradigm: For each ball we are given a choice of $d$ bins, according to $d$ hash functions $h_1,\dots,h_d$ and we place the ball in the least loaded of these bins breaking ties arbitrarily. Our interest is in the number of balls in the fullest bin after all $m$ balls have been placed. 

Azar et al.~[STOC'94] proved that when $m=O(n)$ and when the hash functions are fully random the maximum load is at most $\frac{\lg \lg n }{\lg d}+O(1)$ whp (i.e. with probability $1-O(n^{-\gamma})$ for any choice of $\gamma$). 

In this paper we suppose that $h_1,\dots,h_d$ are simple tabulation hash functions which are simple to implement and can be evaluated in constant time. Generalising a result by Dahlgaard et al~[SODA'16] we show that for an arbitrary constant $d\geq 2$ the maximum load is $O(\lg \lg n)$ whp, and that expected maximum load is at most $\frac{\lg \lg n}{\lg d}+O(1)$.  We further show that by using a simple tie-breaking algorithm introduced by V\"ocking~[J.ACM'03] the expected maximum load drops to $\frac{\lg \lg n}{d\lg \varphi_d}+O(1)$ where $\varphi_d$ is the rate of growth of the $d$-ary Fibonacci numbers. Both of these expected bounds match those of the fully random setting.

The analysis by Dahlgaard et al. relies on a proof by P\u{a}tra\c{s}cu and Thorup~[J.ACM'11] concerning the use of simple tabulation for cuckoo hashing. We require a generalisation to $d>2$ hash functions, but the original proof is an 8-page tour de force of ad-hoc arguments that do not appear to generalise. Our main technical contribution is a shorter, simpler and more accessible proof of the result by P\u{a}tra\c{s}cu and Thorup, where the relevant parts generalise nicely to the analysis of $d$ choices.
 \end{abstract}

\section{Introduction}
Suppose that we are to place $m=O(n)$ balls sequentially into $n$ bins. If the positions of the balls are chosen independently and uniformly at random it is well-known that the maximum load of any bin is\footnote{All logarithms in this paper are binary.} $\Theta(\log n/\log \log n)$ whp (i.e.~with probability $1-O(n^{-\gamma})$ for arbitrarily large fixed $\gamma$). See for example~\cite{Gon} for a precise analysis. 

Another allocation scheme is the \textbf{$d$-choice paradigm} (also called the $d$-choice balanced allocation scheme) first studied by Azar et al.~\cite{Azar}: The  balls are inserted sequentially by for each ball choosing $d$ bins, according to $d$ hash functions $h_1,\dots,h_d$ and placing the ball  in the one of these $d$ bins with the least load, breaking ties arbitrarily. Azar et al.~\cite{Azar} showed that using independent and fully random hash functions the maximum load surprisingly drops to at most $\frac{\log \log n}{\log d}+O(1)$ whp. This result triggered an extensive study of this and related types of load balancing schemes. Currently the paper by Azar et al. has more than 700 citations by theoreticians and practitioners alike. The reader is referred to the text book~\cite{Mit} or the recent survey~\cite{Wieder} for thorough discussions. Applications are numerous and  are surveyed in~\cite{Mitz2,Mitz3}. 

An interesting variant was introduced by V\"ocking~\cite{Voc}. Here the bins are divided into $d$ groups each of size $g=n/d$ and for each ball we choose a single bin from each group. The balls are inserted using the $d$-choice paradigm but in case of ties we always choose the leftmost of the relevant bins i.e.~the one in the group of the smalles index. V\"ocking proved that in this case the maximum load drops further to $\frac{\log \log n}{d \log \varphi_d}+O(1)$ whp.

In this paper we study the use of simple tabulation hashing in the load balancing schemes by Azar et  al. and by V\"ocking.

\subsection{Simple tabulation hashing}
Recall that a hash function $h$ is a map from a key universe $U$ to a range $R$ chosen with respect to some probability distribution on $R^U$. If the distribution is uniform we say that $h$ is fully random but we may impose any probability distribution on $R^U$. 

Simple tabulation hashing was first introduced by Zobrist~\cite{Zobrist}. In simple tabulation hashing $U=[u]=\{0,1,\dots,u-1\}$ and $R=[2^r]$ for some $r$. We identify  $R$ with the $\Z_2$-vector space $(\Z_2)^r$. The keys $x\in U$ are viewed as vectors consisting of $c> 1$ characters $x=(x[0],\dots,x[c-1])$ with each $x[i]\in \Sigma\overset{def}{=}[u^{1/c}]$. We always assume that $c=O(1)$. The simple tabulation hash function $h$ is defined by
$$
h(x)=\bigoplus_{i=0}^{c-1}h_i(x[i])
$$
where $h_0,\dots,h_{c-1}:\Sigma \to R$ are chosen independently  and uniformly at random from $R^\Sigma$. Here $\oplus$ denotes the addition in $R$ which can in turn be interpreted as the bit-wise XOR of the elements $h_i(x[i])$ when viewed as bit-strings of length $r$. 

Simple tabulation is trivial to implement, and very efficient as the character tables $h_0,\dots,h_{c-1}$ fit in fast cache. 
P\u{a}tra\c{s}cu and Thorup~\cite{Pat} considered the hashing of 32-bit keys divided into 4 8-bit characters, and found it to be as fast as two 64-bit multiplications. On computers with larger cache, it may be faster to use 16-bit characters. We note that the $c$ character table lookups can be done in parallel and that character tables are never changed once initialised.

In the $d$-choice paradigm, it is very convenient that all the output bits of simple tabulation are completely independent (the $j$th bit of $h(x)$ is the XOR of the $j$th bit of each $h_i(x[i])$). Using $(dr)$-bit hash values, can therefore be viewed as using $d$ independent $r$-bit hash values, and the $d$ choices can thus be computed using a single simple tabulation hash function and therefore only $c$ lookups.

\subsection{Main results}
We will study the maximum load when the elements of a fixed set $X\subset U$ with $|X|=m$ are distributed into $d$ groups of bins $G_1,\dots,G_d$ each of size $g=n/d$ using the $d$-choice paradigm with independent simple tabulation hash functions $h_1,\dots,h_d:U \to [n/d]$. The $d$ choices thus consist of a single bin from each group as in the scheme by V\"ocking but we may identify the codomain of $h_i$ with $[n/d]\times \{i\}$ and think of all $h_i$ as mapping to the same set of bins $[n/d]\times [d]$ as in the scheme by Azar et al.

Dahlgaard et al.~\cite{Dahl1} analysed the case $d=2$. They proved that if $m=O(n)$ balls are distributed into two tables each  consisting of $n/2$ bins according to the two choice paradigm using two independently chosen simple tabulation hash functions, the maximum load of any bin is $O(\log \log n)$ whp. For $k=O(1)$ they further provided an example where the maximum load is at least $\lfloor k^{c-1}/2\rfloor \log \log n-O(1)$ with probability $\Omega(n^{-2(k-1)(c-1)})$. Their example generalises to arbitrary fixed $d\geq 2$ so we cannot hope for a maximum load of $(1+o(1))\frac{\log \log n}{\log d}$ or even $100\times \log \log n$ whp when $d$ is constant. However, as we show in Appendix~\ref{finishingapp}, their result implies that even with $d=O(1)$ choices the maximum load is $O(\log \log n)$ whp.

Dahlgaard et al. also proved that the expected maximum load is at most  $\log \log n+O(1)$ when $d=2$. We prove the following result which generalises this to arbitrary $d=O(1)$. 
\begin{theorem}\label{mythm}
Let $d>1$ be a fixed constant. Assume $m=O(n)$ balls are distributed into $d$ tables each of size $n/d$ according to the $d$-choice paradigm using $d$ independent simple tabulation hash functions $h_1,\dots,h_d:U \to [n/d]$. Then the expected maximum load is at most $\frac{\log \log n}{\log d}+O(1)$.
\end{theorem}

When in the $d$-choice paradigm we sometimes encounter ties when placing a ball --- several bins among the $d$ choices may have the same minimum load. As observed by V\"ocking~\cite{Voc} the choice of tie breaking algorithm is of subtle importance to the maximum load. In the fully random setting, he showed that if we use the \textbf{Always-Go-Left algorithm} which in case of ties places the ball in the leftmost of the relevant bins, i.e.~in the bin in the group of the smallest index, the maximum load drops to $\frac{\log \log n}{d\log\varphi_d}+O(1)$ whp. Here $\varphi_d$ is the \emph{unique} positive real solution to the equation $x^d=x^{d-1}+\dots+x+1$. We prove that his result holds in expectation when using simple tabulation hashing.

\begin{theorem}\label{alwaysleftthm}
Suppose that we in the setting of Theorem~\ref{mythm} use the Always-Go-Left algorithm for tie-breaking. Then the expected maximum load of any bin is at most $\frac{\log \log n}{d\log\varphi_d}+O(1)$.
\end{theorem}

Note that $\varphi_d$ is the rate of growth of the so called $d$-ary Fibonacci numbers for example defined by $F_d(k)=0$ for $k\leq 0$, $F_d(1)=1$ and finally $F_d(k)=F_d(k-1)+\dots+F_d(k-d)$ when $k>1$. With this definition $\varphi_d=\lim_{k\to \infty} \sqrt[k]{F_d(k)}$. It is easy to check that $(\varphi_{d})_{d>1}$ is an increasing sequence converging to 2.

\subsection{Technical contributions}

In proving Theorem~\ref{mythm} we would ideally like to follow the approach by Dahlgaard et al.~\cite{Dahl1} for the case $d=2$ as close as possible. They show that if some bin gets load $k+1$ then either the hash graph (informally, the $d$-uniform hypergraph with an edge $\{(h_i(x),i))\}_{1\leq i \leq d}$ for each $x\in X$) contains a subgraph of size $O(k)$ with more edges than nodes or a certain kind of ``witness tree'' $T_k$. They then bound the probability that either of these events occur when $k=\log \log n+r$ for some sufficiently large constant $r$. Putting $k=\frac{\log \log n}{\log d}+r$ for a sufficiently large constant $r$ we similarly have three tasks:

\begin{enumerate}
\item[(1)] Define the $d$-ary witness trees and argue that if some bin gets load $k+1$ then either \textbf{(A)}: the hash graph contains a such, or \textbf{(B)}: it contains a subgraph $G=(V,E)$ of size $O(k)$ with $|V|\leq (d-1)|E|-1$.
\item[(2)] Bound the probability of \textbf{(A)}.
\item[(3)] Bound the probability of \textbf{(B)}.
\end{enumerate}

Step (1) and (2) require intricate arguments but the techniques are reminiscent to those used by Dahlgaard et al. in~\cite{Dahl1}. It is not surprising that their arguments generalise to our setting and we will postpone our work with step (1) and (2) to the appendices.

Our main technical contribution is our work on step (3) as we now describe. Dealing with step (3) in the case $d=2$ Dahlgaard et al. used the proof by P\u{a}tra\c{s}cu and Thorup~\cite{Pat} of the result below concerning the use of simple tabulation for cuckoo hashing\footnote{Recall that in cuckoo hashing, as introduced by Pagh and Rodler~\cite{Pagh}, we are in the 2-choice paradigm but we require that no two balls collide. However, we are allowed to rearrange the balls at any point and so the feasibility does only depend on the choices of the balls.}. 
 \begin{theorem}[P\u{a}tra\c{s}cu and Thorup~\cite{Pat}]\label{cuckoothm}
Fix $\varepsilon>0$. Let $X\subset U$ be any set of $m$ keys. Let $n$ be such that $n>2(1+\varepsilon)m$. With probability $1- O(n^{-1/3})$ the keys of $X$ can be placed in two tables of size $n/2$ with cuckoo hashing using two independent simple tabulation hash functions $h_0$ and $h_1$. 
\end{theorem}

Unfortunately for us, the original proof of Theorem~\ref{cuckoothm} consists of 8 pages of intricate ad-hoc arguments that do not seem to generalise to the $d$-choice setting. Thus we have had to develop an alternative technique for dealing with step (3) As an extra reward this technique gives a new proof of Theorem~\ref{cuckoothm} which is shorter, simpler and more readable and we believe it to be our main contribution and of independent interest\footnote{We mention in passing that Theorem~\ref{cuckoothm} is best possible: There exists a set $X$ of $m$ keys such that with probability $\Omega(n^{-1/3})$ cuckoo hashing is forced to rehash (see~\cite{Pat}).}.

\subsection{Alternatives}

We have shown that balanced allocation with $d$ choices
with simple tabulation gives the same expected maximum
load as with fully-random hashing. Simple tabulation uses $c$ lookups
in tables of size $u^{1/c}$ and $c-1$ bit-wise XOR. The experiments from 
\cite{Pat}, with $u=2^{32}$ and $c=4$,
indicate this to be about as fast as two multiplications.

Before comparing with alternative hash functions, we note that we may
assume that $u\leq n^2$. If $u$ is larger, we can first apply a
universal hash function~\cite{Weg} from $[u]$
to $[n^2]$. This yields an expected number of ${n\choose 2}/n^2<1/2$
collisions. We can now apply {\em any\/} hash function, e.g., simple
tabulation, to the reduced keys in $[n^2]$. Each of the duplicate keys
can increase the maximum load by at most one, so the expected maximum
load increases by at most $1/2$. If $u=2^w$, we can use the extremely
simple universal hash function from
\cite{DBLP:journals/jal/DietzfelbingerHKP97}, multiplying the key by a
random odd $w$-bit number and performing a right-shift.

Looking for alternative hash functions, it can be checked that
$O(\log n)$-independence suffices to get the same maximum load bounds
as with full randomness even with high probability.  High independence
hash functions were pioneered by Siegel~\cite{Sieg} and the most
efficient construction is the double tabulation of Thorup
\cite{Mik}. It gives independence $u^{\Omega(1/c^2)}$
using space $O(c u^{1/c})$ in time $O(c)$. With $c$ a constant this would suffice for our purposes. However, looking into
the constants suggested in~\cite{Mik}, with 16-bit characters for
32-bit keys, we have 11 times as many character table lookups
with double tabulation as with simple tabulation and we loose the same
factor in space, so this is not nearly as efficient.

Another approach was given by Woelfel~\cite{Woelf}  using the
hash functions he earlier developed with Dietzfelbinger~\cite{Dietz}. 
He analysed V\"ocking's Always-Go-Left algorithm, bounding
the error probability that the maximum load exceeded 
$\frac{\log \log n}{d\log\varphi_d}+O(1)$. Slightly
simplified and translated to match our notation, 
using $d+1$ $k$-independent hash functions
and $d$ lookups in tables of size $n^{2/c}$, the error
probability is $n^{1+o(1)-k/c}$. Recall that we may assume
$n^{2/c}\geq u^{1/c}$, so this matches the space of simple
tabulation with $c$ characters. With, say, $c=4$, he needs 5-independent
hashing to get any non-trivial bound, but the fastest 5-independent hashing
is the tabulation scheme of Thorup and Zhang~\cite{Zhang}, which
according to the experiments in~\cite{Pat} is at least twice as
slow as simple tabulation, and much more complicated to implement.

A final alternative is to compromise with the constant evaluation
time. Reingold et al.~\cite{Rein} have shown that using the hash
functions from~\cite{Celis} yields a maximum load of $O(\log
\log n)$ whp. The functions use $O(\log n \log \log n)$ random bits
and can be evaluated in time $O((\log \log n)^2)$. Very recently Chen
\cite{Chen} used a refinement of the hash family from~\cite{Celis}
giving a maximum load of at most $\frac{\log \log n}{\log d}+O(1)$
whp and $\frac{\log \log n}{d\log \varphi_d}+O(1)$ whp using
the Always-Go-Left algorithm. His functions require $O(\log n\log \log
n)$ random bits and can be evaluated in time $O((\log \log n)^4)$. 
We are not so concerned with the number of random bits. Our main interest in simple tabulation is in the constant evaluation time with a very low constant.

\subsection{Structure of the paper}
In Section~\ref{Prel} we provide a few preliminaries for the proofs of our main results. In Section~\ref{Cuckoosec} we deal with step (3) described under \emph{Technical contributions}. To provide some intuition we first provide the new proof of Theorem~\ref{cuckoothm}. Afterwards, we show how to proceed for general $d$. In Appendix~\ref{Impliapp} we show how to complete step (1) In Appendix~\ref{Tree1app} and Appendix~\ref{Tree2app} we complete step (2) Finally we show how to complete the proof of Theorem~\ref{mythm} and Theorem~\ref{alwaysleftthm} in Appendix~\ref{finishingapp}. In Appendix~\ref{openapp} we mention a few open problems.
\section{Preliminaries}\label{Prel}
First, recall the definition of a hypergraph:
\begin{definition}
A \textbf{hypergraph} is a pair $G=(V,E)$ where $V$ is a set and $E$ is a multiset consisting of elements from $\mathcal{P}(V)$. The elements of $V$ are called \textbf{vertices} and the elements of $E$ are called \textbf{edges}. We say that $G$ is \textbf{$d$-uniform} if $|e|=d$ for all $e\in E$.
\end{definition}
When using the $d$-choice paradigm to distribute a set of keys $X$ there is a natural $d$-uniform hypergraph associated with the keys of $X$.
\begin{definition}
Given a set of keys $X\subset U$ the \textbf{hash graph} is the $d$-uniform hypergraph on $[n/d]\times [d]$ with an edge $\{(h_1(x),1),\dots,(h_d(x),d)\}$ for each $x\in X$.
\end{definition}
When working with the hash graph we will hardly ever distinguish between a key $x$ and the corresponding edge, since it is tedious to write $\{(h_i(x),i)\}_{1\leq i \leq d}$. Statements such as ``$P=(x_1,\dots,x_t)$ is a path'' or  ``The keys $x_1$ and $x_2$ are adjacent in the hash graph'' are examples of this abuse of notation. 

Now we discuss the independence of simple tabulation. First recall that a \textbf{position character} is an element $(j,\alpha)\in [c]\times \Sigma$. With this definition a key $x\in U$ can be viewed as the set of position characters $\{(i,x[i])\}_{i=0}^{c-1}$ but it is sensible to define $h(S)=\bigoplus_{i=1}^k h_{j_i}(\alpha_i)$ for any set $S=\{(j_1,\alpha_1),\dots,(j_k,\alpha_k)\}$ of position characters.

In the classical notion of independence of Carter and Wegman~\cite{Weg} simple tabulation is not even 4-independent. In fact, the keys $(a_0,b_0),(a_0,b_1),(a_1,b_0)$ and $(a_1,b_1)$ are dependent, the issue being that each position character appears an even number of times and so the bitwise XOR of the hash values will be the zero string. As proved by Thorup and Zhang~\cite{Zhang} this property in a sense characterises dependence of keys.
\begin{lemma}[Thorup and Zhang~\cite{Zhang}]\label{depchar}
The keys $x_1,\dots,x_k\in U$ are dependent if and only if there exists a non-empty subset $I\subset \{1,\dots,k\}$ such that each position character in $(x_i)_{i\in I}$ appears an even number of times. In this case we have that $\bigoplus_{i\in I}h(x_i)=0$.
\end{lemma}

When each position character appears an even number of times in $(x_i)_{i\in I}$ we will write $\bigoplus_{i\in I}x_i=\emptyset$ which is natural when we think of a key as a set of position characters and $\oplus$ as the symmetric difference. As shown by Dahlgaard et al.~\cite{Dahl2} the characterisation in Lemma~\ref{depchar} can be used to bound the independence of simple tabulation.
\begin{lemma}[Dahlgaard et al.~\cite{Dahl2}]\label{deplemma} Let $A_1,\dots,A_{2t}\subset U$. The number of $2t$-tuples $(x_1,\dots,x_{2t})\in A_1\times \cdots \times A_{2t}$ such that $x_1\oplus \cdots \oplus x_{2t}=\emptyset$ is at most\footnote{Recall the double factorial notation: If $a$ is a positive integer we write $a!!$ for the product of all the positive integers between $1$ and $a$ that have the same parity as $a$.} $((2t-1)!!)^c\prod_{i=1}^{2t}\sqrt{|A_i|}$.
\end{lemma}
This lemma will be of extreme importance to us. For completeness we provide proofs of both Lemma~\ref{depchar} and Lemma~\ref{deplemma} in Appendix~\ref{Indapp}.
\section{Cuckoo hashing and generalisations}\label{Cuckoosec}
The following result is a key ingredient in the proofs of Theorem~\ref{mythm} and Theorem~\ref{alwaysleftthm}.
\begin{theorem}\label{tightthm}
Suppose that we are in the setting of Theorem~\ref{mythm} i.e. $d>1$ is a fixed constant, $X\subset U$ with $|X|=m=O(n)$ and $h_1,\dots,h_d:U\to [n/d]$ are independent simple tabulation hash functions. The probability that the hash graph contains a subgraph $G=(V,E)$  of size $|E|=O(\log \log n)$ with $|V|\leq (d-1)|E|-1$ is at most $n^{-1/3+o(1)}$. 
\end{theorem}
Before giving the full proof however we provide the new proof of Theorem~\ref{cuckoothm} which is more readable and illustrates nearly all the main ideas. 
\begin{figure}
  \centering
    \includegraphics[width=0.6\textwidth]{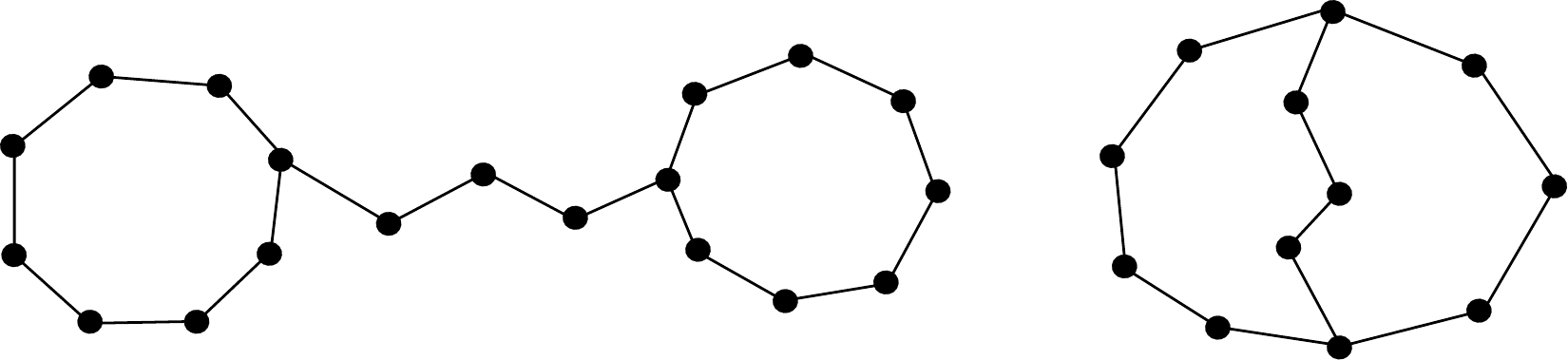}
      \caption{Double cycles - the minimal obstructions for cuckoo hashing.}\label{simpledoublecycles}
      \label{doublecycles}
\end{figure}
\begin{proof}[Proof of Theorem~\ref{cuckoothm}]
It is well known that cuckoo hashing is possible if and only if the hash graph contains no subgraph with more edges than nodes. A minimal such graph is called a \textbf{double cycle} and consists of two cycles connected by a path or two vertices connected by three disjoint paths (see Figure~\ref{doublecycles}). Hence, it suffices to bound the probability that the hash graph contains a double cycle by $O(n^{-1/3})$.

We denote by $g$ the number of bins in each of the two groups. Thus in this setting $g=n/2\geq(1+\varepsilon)m$. First of all, we argue that we may assume that the hash graph contains no trail of length at least $\ell=\frac{4}{3}\frac{\log n}{\log(1+\varepsilon)}$ consisting of \emph{independent}. Indeed, the keys of a such can be chosen in at most $m^\ell$ ways and since we require $\ell-1$ equations of the form $h_i(x)=h_i(y)$, $i\in \{1,2\}$ to be satisfied and since these events are independent the probability that the hash graph contains such a trail is by a union bound at most
\begin{align*}
\frac{2m^{\ell}}{g^{\ell-1}}\leq \frac{n}{(1+\varepsilon)^\ell}=n^{-1/3}.
\end{align*}

Now we return to the double cycles. Let $A_{\ell}$ denote the event that the hash graph contains a double cycle of size $\ell$ consisting of \emph{independent} keys. The graph structure of a such can be chosen in $O(\ell^2)$ ways and the keys (including their positions) in at most $m^\ell$ ways. Since there are $\ell+1$ equations of the form $h_i(x)=h_i(y)$, $i\in \{1,2\}$ to be satisfied the probability that the hash graph contains a double cycle consisting of independent keys is at most
\begin{align*}
\sum_{\ell=3}^m \P(A_\ell) =O\left(\sum_{\ell=3}^m\ell^2 \frac{m^{\ell}}{g^{\ell+1}}\right) =  O\left(\frac{1}{n}\sum_{\ell=3}^m\frac{2\ell^2}{\left(1+\varepsilon \right)^{\ell}}\right) =O( n^{-1}).
\end{align*}
The argument above is the same as in the fully random setting. We now turn to the issue of dependencies in the double cycle starting with the following definition.
\begin{definition}
We say that a graph is a \textbf{trident} if it consists  of three paths $P_1,P_2,P_3$ of non-zero lengths meeting at a single vertex $v$. (see the non-black part of Figure~\ref{tridents}).

We say that a graph is a \textbf{lasso} if it consists of a path that has one end attached to a cycle  (see the non-black part of Figure~\ref{tridents}).
\end{definition}
We claim that in any double cycle $D$ consisting of \emph{dependent} keys we can find one of the following structures (see Figure~\ref{tridents}):
\begin{itemize}
\item \textbf{S1:} A lasso $L$ consisting of independent keys together with a key $x$ not on $L$ and incident to the degree 1 vertex of $L$ such that $x$ is dependent on the keys of $L$.
\item \textbf{S2:} A trident $T$ consisting of independent keys together with $3$ (not necessarily distinct) keys $x,y,z$ not on $T$ but each dependent on the keys of $T$ and incident to the $3$ vertices of degree $1$ on $T$
\end{itemize}
\begin{figure}
\begin{center}
\def\svgwidth{0.8\linewidth}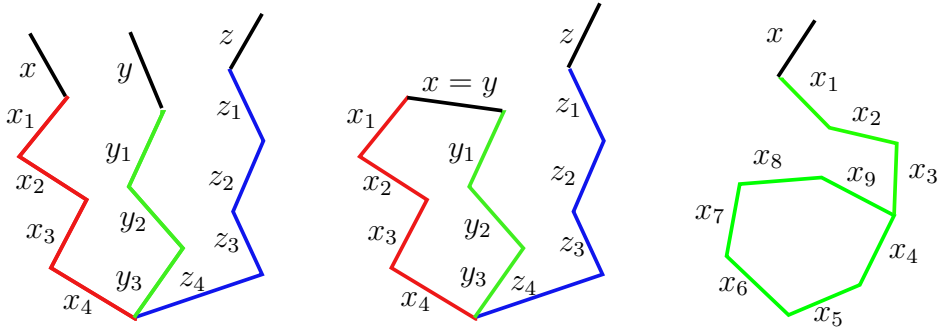
\caption{Non-black edges: Two tridents and a lasso. Black edges: Keys that are each dependent on the set of coloured keys.}
\label{tridents}
\end{center}
\end{figure}

To see this suppose first that one of the cycles $C$ of $D$ consists of independent keys. In this case any maximal lasso of independent keys in $D$ containing the edges of $C$ is an $S_1$.

On the other hand if all cycles contained in $D$ consist of dependent keys we pick a vertex of $D$ of degree at least $3$ and $3$ incident edges. These 3 edges form an independent trident (simple tabulation is 3-independent) and any maximal independent trident contained in $D$ and containing these edges forms an $S_2$. 

Our final step is thus to show that the probability that these structures appear in the hash graph is $O(n^{-1/3})$
\paragraph*{The lasso ($S_1$):} 
Since the edges of the lasso form an independent trail it by the initial observation suffices to bound the probability that the hash graph contains an $S_1$ of size $\ell$ for any $\ell=O(\log n)$.

Fix the size $\ell$ of the lasso. The number of ways to choose the structure of the lasso is $\ell-2<\ell$. Denote the set of independent keys of the lasso by $S=\{x_1,\dots,x_\ell\}$ and let $x$ be the dependent key in $S_1$. By Lemma~\ref{depchar} we may write $x=\bigoplus_{i\in I} x_i$ for some $I\subset \{1,\dots,\ell\}$. Fix the size $|I|=t\geq 3$  (which is necessarily odd). By Lemma~\ref{deplemma} the number of ways to choose the keys of $(x_i)_{i\in I}$ (including their order) is at most $(t!!)^cm^{(t+1)/2}$ and the number of ways to choose their positions in the lasso is $\binom{\ell}{t}$. The number of ways to choose the remaining keys of $S$ is trivially bounded by $m^{\ell-t}$ and the probability that the choice of independent keys hash to the correct positions in the lasso is at most $2/g^\ell$. By a union bound the probability that the hash graph contains an $S_1$ for fixed values of $\ell$ and $t$ is at most
\begin{align*}
\ell (t!!)^c m^{(t+1)/2}m^{\ell-t} \binom{\ell}{t}\frac{2}{g^{\ell}}.
\end{align*}
This is maximised for $t=3$. In fact, when $\ell \leq m^{1/(c+2)}$ and $t\leq \ell-2$ we have that
\begin{align*}
\frac{ ((t+2)!!)^c m^{(t+3)/2}m^{\ell-t-2}\binom{\ell}{t+2}}{ (t!!)^c m^{(t+1)/2}m^{\ell-t} \binom{\ell}{t}} =\frac{(t+2)^c}{m}\frac{\binom{\ell-t}{2}}{\binom{t+2}{2}}\leq \frac{\ell^{c+2}}{m}\leq 1.
\end{align*}
Thus the probability that the hash graph contains an $S_1$ of size $O(\log n)$ is at most 
\begin{align*}
\sum_{\ell=3}^{O(\log n)}\sum_{t=3}^\ell \ell 3^c \binom{\ell}{3} \frac{2m^{\ell-1}}{g^\ell}=O\left( \sum_{\ell=3}^{O(\log n)} \frac{\ell^5}{n(1+\varepsilon)^{\ell-1}}\right)=O( n^{-1}).
\end{align*}
\paragraph*{The trident ($S_2$):}
Fix the size $\ell$ of the trident. The number of ways to choose the structure of the trident is bounded by $\ell^2$ (once we choose the lengths of two of the paths the length of the third becomes fixed). Let $P_1=(x_1,\dots,x_{t_1})$, $P_2=(y_1,\dots,y_{t_2})$ and $P_3=(z_1,\dots,z_{t_3})$ be the three paths of the trident meeting in $x_{t_1}\cap y_{t_2}\cap z_{t_3}$. As before we may assume that each has length $O(\log n)$. Let $S$ denote the keys of the trident and enumerate $S=\{w_1,\dots,w_\ell\}$ in some order. Write $x=\bigoplus_{i\in I} w_{i}$, $y=\bigoplus_{j\in J} w_{j}$ and $z=\bigoplus_{k\in K} w_{j}$ for some $I,J,K\subset \{1,\dots,\ell\}$. By a proof almost identical to that given for the lasso we may assume that $|I|=|J|=|K|=3$. Indeed, if for example $|I|\geq 5$ we by Lemma~\ref{deplemma} save a factor of nearly $m^2$ when choosing the key of $S$ and this makes up for the fact that the trident contains no cycles and hence that the probability of a fixed set of independent keys hashing to it is a factor of $g$ larger.

The next observation is that we may assume that  $|I\cap J|,|J\cap K|,|K\cap I|\geq 2$. Again the argument is of the same flavour as the one given above. If for example $|I \cap J|=1$ we by  an application of Lemma~\ref{deplemma} obtain that the number of ways to choose the keys of $(w_i)_{i\in I}$ is $O(m^2)$. Conditioned on this, the number of ways to choose the keys $(w_j)_{j\in J}$ is $O(m^{3/2})$ by another application of Lemma~\ref{deplemma} with one of the $A_i$'s a singleton. Thus we save a factor of $m^{3/2}$ when choosing the keys of $S$ which will again suffice. The bound gets even better when $|I\cap J|=0$ where we save a factor of $m^2$.

Suppose now that $x_1$ is not a summand of $\bigoplus_{i\in I} w_i$. Write $x=w_a\oplus w_b\oplus w_c$ and let $A$ be the event that the independent keys of $S$ hash to the trident (with the equation involving $x_1$ and $x_2$ being $h_2(x_1)=h_2(x_2)$ without loss of generality). Then $\P(A)=\frac{1}{g^{\ell-1}}$. We observe that 
\begin{align*}
\P(h_1(x)=h_1(x_1)\, | \, A)=\P(h_1(x_1)=h_1(w_a)\oplus h_1(w_b) \oplus h_1(w_c)\, | \, A)=g^{-1}
\end{align*}
since $A$ is a conjunction of events of the form $\{h_i(w)=h_i(w')\}$ none of them involving $h_1(x_1)$\footnote{If $x_1=w_a$, say, we don't necessarily get the probability $g^{-1}$. In this case the probability is $\P (h_1(w_b)=h_1(w_c) \, | \, A)$ and the event $\{h(w_b)=h(w_c) \}$ might actually be included in $A$ in which case the probability is $1$. This can of course only happen if the keys $w_b$ and $w_c$ are \emph{adjacent} in the trident so we could impose even further restrictions on the dependencies in $S_2$. }.  A  union bound then gives that the probability that this can happen is at most
\begin{align*}
 \sum_{\ell=3}^{O(\log n)}\ell^2 \binom{\ell}{3}(3!!)^cm^2 m^{\ell-3}\left(\frac{1}{g}\right)^\ell=O\left( \frac{1}{n}\sum_{\ell=3}^\infty \frac{\ell^5}{(1+\varepsilon)^{\ell-1}}\right)=O(n^{-1}).
\end{align*}
Thus we may assume that $x_1$ is a summand of $\bigoplus_{i\in I}w_i$ and by similar arguments that $y_1$ is a summand of $\bigoplus_{j\in J} w_{j}$ and that $z_1$ is a summand of $\bigoplus_{k\in K} w_{k}$.

To complete the proof we need one final observation. We can define an equivalence relation on $X\times X$ by $(a,b)\sim (c,d)$ if $a\oplus b=c\oplus d$. Denote by $\mathcal{C}=\{C_1,\dots,C_r\}$ the set of equivalence classes. One of them, say $C_1$, consists of the elements $(x,x)_{x\in X}$. We will say that the equivalence class $C_i$ is \textbf{large} if $|C_i|\geq m^{2/3}$ and \textbf{small} otherwise. Note that 
\begin{align*}
\sum_{i=1}^r |C_i|^2=|\{(a,b,c,d)\in X^4:a\oplus b\oplus c \oplus d=\emptyset\}|\leq 3^cm^2
\end{align*}
by Lemma~\ref{deplemma}. In particular  the number of large equivalence classes is at most $3^cm^{2/3}$. 

If $h$ is a simple tabulation hash function we can well-define a map $\tilde h:\mathcal {C} \to R$ by $\tilde h(a,b)=h(a)\oplus h(b)$. Since the number of large equivalence classes is $O(m^{2/3})$ the probability that $\tilde h_i(C)=0$ for some large $C\in \mathcal{C}\backslash \{C_1\}$ and some $i\in \{1,2\}$ is  $O(m^{2/3}/n)=O(n^{-1/3})$ and we may thus assume this does not happen. 

In particular, we may assume that $(x,x_1)$, $(y,y_1)$ and $(z,z_1)$ each represent small equivalence classes as they are adjacent in the hash graph. Now suppose that $y_1$ is not a summand in $x=\bigoplus_{i\in I} w_{i}$. The number of ways to pick $(x_i)_{i\in I}$ is at most $3^cm^2$ by Lemma~\ref{deplemma}. By doing so we fix the equivalence class of $(y,y_1)$ but not $y_1$ so conditioned on this the number of ways to pick $(y_j)_{j\in J}$ is at most $m^{2/3}$. The number of ways to choose the remaining keys is bounded by $m^{\ell-4}$ and a union bound gives that the probability of having such a trident is at most
\begin{align*}
 \sum_{\ell=3}^{O(\log n)} \ell^23\binom{\ell}{2}3^cm^2m^{2/3}m^{\ell-4}\left(\frac{1}{g}\right)^{\ell-1} =O\left( n^{-1/3} \sum_{\ell=3}^\infty \frac{\ell^4}{(1+\varepsilon)^{\ell-4/3}}\right)=O( n^{-1/3}),
\end{align*}
which suffices. 

We may thus assume that $y_1$ is a summand in $\bigoplus_{i\in I} w_{i}$ and by an identical argument that $z_1$ is a summand in $\bigoplus_{i\in I} w_{i}$ and hence $x=x_1\oplus y_1 \oplus z_1$. But the same arguments apply to $y$ and $z$ reducing to the case when $x=y=z=x_1\oplus y_1 \oplus z_1$ which is clearly impossible.
\end{proof}

\subsection{Proving Theorem~\ref{tightthm}}
Now we will explain how to prove Theorem~\ref{tightthm} proceeding much like we did for Theorem~\ref{cuckoothm}. Let us say that a $d$-uniform hypergraph $G=(V,E)$ is \textbf{tight} if $|V|\leq (d-1)|E|-1$. With this terminology Theorem~\ref{tightthm}  states that the probability that the hash graph contains a tight subgraph of size $O(\log \log n)$ is at most $n^{-1/3+o(1)}$. It clearly suffices to bound the probability of the existence of a \emph{connected} tight subgraph of size $O(\log \log n)$.

We start with the following two lemmas. The counterparts in the proof of Theorem~\ref{cuckoothm} are the bounds on the probability of respectively an independent double cycle and an independent lasso with a dependent key attached.
\begin{lemma}\label{lemser1}
Let $A_1$ denote the event that the hash graph contains a tight subgraph $G=(V,E)$ of size $O(\log \log n)$ consisting of independent keys. Then $\P(A_1)\leq n^{-1+o(1)}$.
\end{lemma}
\begin{proof}
Let $\ell=|E|$ be fixed. The number of ways to choose the keys of $E$ is trivially bounded by $m^\ell$ and the number of ways to choose the set of nodes $V$ in the hash graph is $\binom{n}{(d-1)\ell-1}$. For such a choice of nodes let $a_i$ denote the number of nodes of $V$ in the $i$'th group. The probability that one of the keys hash to $V$ is then
\begin{align*}
\prod_{i=1}^d \frac{da_i}{n} \leq\left(\frac{a_1+\dots+a_d}{n} \right)^d\leq\left( \frac{d\ell}{n}\right)^d.
\end{align*}
By the independence of the keys and a union bound we thus have that 
\begin{align*}
\P(A_1)\leq \sum_{\ell=2}^{O(\log \log n)}m^\ell \binom{n}{(d-1)\ell-1}\left( \frac{d\ell}{n}\right)^{d\ell}\leq \sum_{\ell=2}^{O(\log \log n)}\frac{1}{n} \left(\frac{m}{n}\right)^{\ell}(d\ell)^{d\ell}=n^{-1+o(1)},
\end{align*}
as desired.
\end{proof}
\begin{lemma}\label{lemser2}
Let $A_2$ be the event that the hash graph contains a subgraph $G=(V,E)$ with $|V|\leq (d-1)|E|$ and $|E|=O(\log \log n)$ such that the keys of $E$ are independent but such that there exists a key $y\notin E$ dependent on the keys of $E$. Then $\P(A_2)\leq n^{-1+o(1)}$.
\end{lemma}
\begin{proof}
Let $|E|=\ell$ be fixed and write $E=\{x_1,\dots,x_\ell\}$ . We want to bound the number of ways to choose the keys of $E$.
By Lemma~\ref{depchar}, $y=\bigoplus_{i \in I} x_i$ for some $I\subset \{1,\dots,\ell\}$ with $|I|=r$  for some odd $r\geq 3$. Let $r$ be fixed for now. Using Lemma~\ref{deplemma}, we see that the number of ways to choose the keys of $E$ is no more than $(r!!)^cm^{\frac{r+1}{2}}m^{\ell-r}$. For fixed $\ell$ and $r$ the probability is thus bounded by
\begin{align*}
 (r!!)^cm^{\ell-\frac{r-1}{2}}\binom{n}{\ell(d-1)} \left( \frac{d\ell}{n}\right)^{d\ell}= n^{-1+o(1)}
\end{align*}
and a union bound over all $\ell=O(\log \log n)$ and $r\leq \ell$ suffices.
\end{proof}
We now generalise the notion of a double cycle starting with the following definition.
\begin{definition}
Let $G=(V,E)$ be a $d$-uniform hypergraph. We say that a sequence of edges $P=(e_1,\dots,e_t)$ of $G$ is a \textbf{path} if $|e_i\cap e_{i+1}|=1$ for $1\leq i \leq t-1 $ and $e_i\cap e_j=\emptyset$ when $i<j-1$.

We say that $C=(e_1,\dots,e_t)$ is a \textbf{cycle} if $t\geq 3$, $|e_i\cap e_{i+1}|=1$ for all $i \pmod {t}$ and $e_i\cap e_j=\emptyset$ when $i\neq j\pm 1 \pmod{t}$.
\end{definition}
Next comes the natural extension of the definition of double cycles to $d$-uniform hypergraphs.
\begin{definition}\label{doublecycdef}
A $d$-uniform hypergraph $G$ is called a \textbf{double cycle} if it has either of the following forms (see Figure~\ref{doubcyc}).
\begin{itemize}
\item \textbf{D1:} It consists of of two vertex disjoint cycles $C_1$ and $C_2$ connected by a path $P=(x_1,\dots,x_t)$ such that $|x_1\cap V(C_1)|=|x_t\cap V(C_2)|=1$ and $x_{i+1}\cap V(C_1)=x_{i}\cap V(C_2)=\emptyset$ for $1\leq i \leq t-1$ . We also allow $P$ to have zero length and $|V(C_1)\cap V(C_2)|=1$. 
\item \textbf{D2:} It consist of a cycle $C$ and a path $P=(x_1,\dots,x_t)$ of length $t\geq 2$ such that $|x_1\cap V(C)|=|x_t\cap V(C)|=1$ and $x_i\cap V(C)=\emptyset$ for $2\leq i \leq t-1$. We also allow $t=1$ and $|x_1\cap C|=2$. 
\end{itemize}
\end{definition}
\begin{figure}
  \centering
  \def\svgwidth{0.7\linewidth}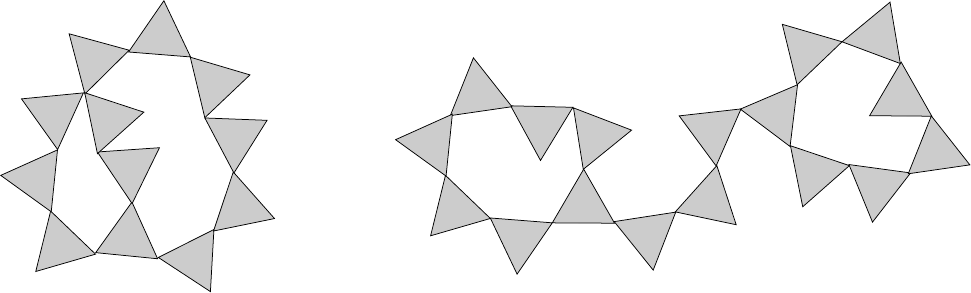
      \caption{Double cycles in the case $d=3$. The triangles represent edges of the graph and the corners represent the vertices.}
        \label{doubcyc}
\end{figure}
Note that a double cycle always has $|V|=(d-1)|E|-1$.

Now assume that the hash graph contains a connected tight subgraph $G=(V,E)$ of size $O(\log \log n)$ but that neither of the events of Lemma~\ref{lemser1} and~\ref{lemser2} has occurred. In particular no two edges $e_1,e_2$ of $G$ has $|e_1\cap e_2|\geq 2$ and no cycle consists of independent keys.

It is easy to check that under this assumption $G$ contains at least two cycles. Now pick a cycle $C_1$ of least possible length. Since simple tabulation is $3$-independent the cycle consists of at least $4$ edges. If there exists an edge $x$ not part of $C_1$ with $|x\cap V(C_1)|=2$ we get a double cycle of type $D_2$. If $|x\cap V(C_1)|\geq 3$ we can use $x$ to obtain a shorter cycle than $C_1$ which is a contradiction\footnote{Here we use that the length of $C_1$ is at least 4. If $C_1$ has length $t$ the fact that $x$ contains three nodes of $C_1$  only guarantees a cycle of length at most $3+\lfloor \frac{t-3}{3} \rfloor$.}. Using this observation we see that if there is a cycle $C_2\neq C_1$ such that $|V(C_1)\cap V(C_2)|\geq 2$ then we can find a  $D_2$ in the hash graph. Thus we may assume that any cycle $C_2\neq C_1$ satisfies $|V(C_2)\cap V(C_1)|\leq 1$.

Now pick a cycle $C_2$ different from $C_1$ of least possible length. As before we may argue that any edge $x$ not part of $C_2$ satisfies that $|x\cap V(C_2)|\leq 1$. Picking a shortest path connecting $C_1$ and $C_2$ (possibly the length is zero) gives a double cycle of type $D_1$.

Next we define tridents (see the non-grey part of Figure~\ref{tridents2}).
\begin{definition}
We call a $d$-uniform hypergraph $T$ a \textbf{trident} if it consists of paths $P_1=(x_1,\dots,x_{t_1})$, $P_2=(y_1,\dots,y_{t_2})$ and $P_3=(z_1,\dots,z_{t_3})$ of non-zero length such that either:\begin{itemize}
\item There is a vertex $v$ such that $x_{t_1}\cap y_{t_2}\cap z_{t_3}=\{v\}$, $v$ is contained in no other edge of $T$ and no vertex different from $v$ is contained in more than one of the three paths.
\item $P_1$, $P_2$ and $P_3\backslash \{z_{t_3}\} =(z_2,\dots,z_{t_3})$ are vertex disjoint and $(x_1,\dots,x_{t_1},z_{t_3},y_{t_2},\dots,y_1)$ is a path.
\end{itemize}
\end{definition}
\begin{figure}
\centering
    \includegraphics[width=0.85\textwidth]{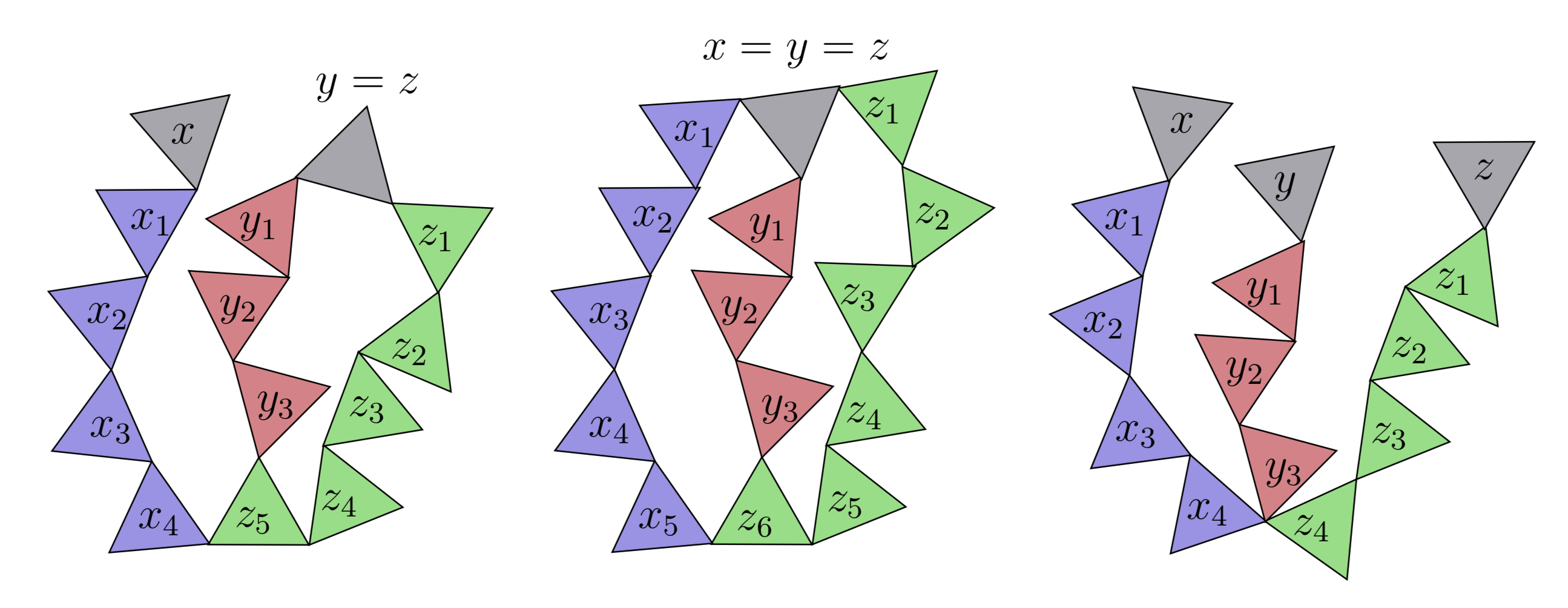}
      \caption{The case $d=3$. Non-grey edges: Tridents. Grey edges: Keys that are each dependent on the set of non-black keys.}
      \label{tridents2}
\end{figure}

Like in the proof of of Theorem~\ref{cuckoothm} the existence of a double cycle not containing a cycle of independent keys implies the existence of the following structure (see Figure~\ref{tridents2}):

\begin{itemize}
\item \textbf{S1}: A trident consisting of three paths $P_1=(x_1,\dots,x_{t_1})$, $P_2=(y_1,\dots,y_{t_2})$ and $P_3=(z_1,\dots,z_{t_3})$ such that the keys of the trident are independent and such that there are, not necessarily distinct, keys $x,y,z$ not in the trident extending the paths $P_1$, $P_2$ and $P_3$ away from their common meeting point such that $x,y$ and $z$ are each dependent on the keys in the trident. 
\end{itemize}

We can bound the probability of this event almost identically to how we proceeded in the proof of Theorem~\ref{cuckoothm}. The only difference is that when making the ultimate reduction to the case where $x=y=z=x_1\oplus y_1\oplus z_1$ this event is in fact possible (see Figure~\ref{tridents2}). In this case however, there are three \emph{different} hash function $h_x,h_y$ and $h_z$ such that $h_x(x_1)=h_x(x)$, $h_y(y_1)=h_y(x)$ and $h_z(z_1)=h_z(x)$. What is the probability that this can happen? The number of ways to choose the keys $(x,x_1,y_1,z_1)$ is at most $3^cm^2$ by Lemma~\ref{deplemma}. The number of ways to choose the hash functions is upper bounded by $d^3$. Since the hash functions $h_1,\dots,h_d$ are independent the probability that this can happen in the hash graph is by a union bound at most
\begin{align*}
d^3 3^cm^2 \left(\frac{d}{n}\right)^3=O( n^{-1})
\end{align*}
which suffices to complete the proof of Theorem~\ref{tightthm}.

\subsection*{Summary}
For now we have spent most of our energy proving Theorem~\ref{tightthm}. At this point it is perhaps not clear to the reader why it is important so let us again highlight the steps to Theorem~\ref{mythm}. First of all let $k=\frac{\log \log n }{\log d}+r$ for $r$ a sufficiently large constant. The steps are:

\begin{enumerate}
\item[(1)] Show that if some bin has load $k$ then either the hash graph contains a tight subgraph of size $O(k)$ or a certain kind of witness tree $T_k$.
\item[(2)] Bound the probability that the hash graph contains a $T_k$ by $O((\log \log n)^{-1})$.
\item[(3)] Bound the probability that the hash graph contains a tight subgraph of size $O(k)$ by $O((\log \log n)^{-1})$.
\end{enumerate}

We can now cross (3) of the list. In fact, we have a much stronger bound. The remaining steps are dealt with in the appendices as described under \emph{Structure of the paper}.

As already mentioned the proofs of all the above steps (except step (3)) are intricate but straightforward generalisations of the methods in~\cite{Dahl1}.

\bibliographystyle{plain}
\bibliography{lipics-v2018-sample-article}

\clearpage
\appendix
\centerline{\huge\bf Appendix}

\section{Implications of having a bin of large load}\label{Impliapp}
Before we start we will introduce some definitions concerning $d$-uniform hypergraphs.  We say that a $d$-uniform hypergraph $G=(V,E)$ is a \textbf{tree} if $G$ is connected and $|V|=(d-1)|E|+1$. We say that $G$ is a \textbf{forest} if the connected component of $G$ are trees.
The following result and its corollary are easily proven.
\begin{lemma}
Let $T=(V,E)$ be a connected $d$-uniform hypergraph. Then $T$ is a tree if and only if $T$ does not contain a cycle  or a pair of distinct edges $e_1,e_2$ with $|e_1\cap e_2|\geq 2$.
\end{lemma}
\begin{corollary}
A connected subgraph of a forest is a tree. 
\end{corollary}
We define a \textbf{rooted tree}  $T=(V,E)$ to be a hypertree where we have fixed a root $v\in V$. We can define the \textbf{depth} of a node to be the length of the shortest path from this vertex to the root. Any edge $e$ in a rooted tree $T$ can be written $e=\{v_1,\dots,v_d\}$ such that for some $\ell$ we have that $v_1$ has depth $\ell$ and $v_2,\dots,v_d$ each has depth $\ell+1$. With this notation we will say that $v_2,\dots,v_d$ are \textbf{children} of $v_1$. We will say that a node $v\in V$ is \textbf{internal} if it has at least one child and that $v$ is a \textbf{leaf} if it has no children. Note finally that for each vertex $w\in V$ we have an induced subtree $T_w$ of $T$ rooted at $w$. If $w$ has depth $\ell$ this tree can be described as  the maximal connected subgraph of $T$ containing $w$ in which each node has depth at least $\ell$. If $w'$ is a node of $T_w$ we will say that $w$ is an \textbf{ancestor} of $w'$ or that $w'$ is a \textbf{descendant} of $w$. 

In the next two subsections we introduce the witnessing trees in the settings of Theorem~\ref{mythm} and~\ref{alwaysleftthm} respectively and show that if some bin has load at least $k$ then either the hash graph will contain a tight subgraph of size $O(k)$ or such a witnessing tree. 
\subsection{The $d$-nomial trees}
To define the witness tree we will need the notion of the $k$'th load graph of a vertex $v$ in the hash graph. It is intuitively a subgraph of the hash graph witnessing how the bin corresponding to $v$ obtained its first $k$ balls. \begin{definition}
Suppose $v$ is a vertex of the hash graph corresponding to a bin of load at least $k$. We recursively define  $L_v(k)$ the  \textbf{$k$'th load graph of $v$} to be the following $d$-uniform hypergraph.
\begin{itemize}
\item If $k=0$ we let $L_v(k)=(\{v\},\emptyset)$.
\item If $k>0$ we let $e$ be the edge corresponding to the $k$'th key landing in $v$. Write $e=\{v_1,\dots,v_d\}$. Then $L_v(k)$ is the graph with 
\begin{align*}
E(L_v(k))=\{e\}\cup \bigcup_{i=1}^d E(L_{v_i}(k-1)) \ \text{ and } \ V(L_v(k))=\bigcup_{e\in E(L_v(k))}e.
\end{align*}
\end{itemize}
\end{definition}

As we are distributing the balls according to the $d$-choice paradigm the definition is sensible.

It should be no surprise that if we know that the $k$'th load graph of a vertex is a tree we can actually describe the structure of that tree. We now describe that tree.

\begin{definition}
A \textbf{$d$-nomial tree} $B_{d,k}$ for $k\geq 0$ is the rooted $d$-uniform hypertree defined recursively as follows: 
\begin{itemize}
\item $B_{d,0}$ is a single node
\item  $B_{d,k}$ is rooted at a vertex $v_1$ and consists of an edge $e=\{v_1,\dots,v_d\}$ such that each $v_i$ is itself a root of a $B_{d,k-1}$.  
\end{itemize}
\end{definition}
Since $d$ will be fixed we will often suppress the $d$ and just write $B_k$. 

\begin{lemma}\label{bincon}
Let $v$ be a vertex of the hash graph for which the corresponding bin has load at least $k$. Suppose that the $k$'th load graph of $V$ is a tree. Then the $k$'th load graph is in fact a $B_k$ rooted at $v$.
\end{lemma}
\begin{proof}
We prove the result by induction on $k$. For $k=0$ the statement is trivial so suppose $k\geq 1$ and that the result holds for smaller values of $k$. If $e$ is the edge corresponding to the $k$'th ball landing in $v$ then the $(k-1)$'st load graphs of the vertices incident to $e$ will by the induction hypothesis be the roots of $d$ disjoint $B_{k-1}$'s. Going back to the definition of the $d$-nomial trees we see that the $k$'th load graph is exactly a $B_k$ rooted at $v$. This completes the proof.
\end{proof}
Suppose that there is a bin of load $k+1$ and consider the $(k+1)$'st load graph $G=(V,E)$ for the node $v$ corresponding to that bin. 
If $|V|=  |E|(d-1)+1$ we know that the load graph is a tree and hence a $B_{k+1}$. If on the other hand $|V|=|E|(d-1)$ it is easy to check that we can remove some edge from $G$ leaving the graph a forest.
Thus, if the $(k+1)$'st load graph has $|V|\geq (d-1)|E|$ removing at most one edge $e$ from it will turn it into a forest. But removing one edge can decrease the load of $v$ by at most one so if we consider the $k$'th load graph of $v$ in $(V,E-\{e\})$ it will be a tree (being a connected subgraph of a forest). By Lemma~\ref{bincon} above we conclude that it will in fact be a $d$-nomial tree $B_k$ rooted at $v$. We summarise this in the following lemma.
\begin{lemma}\label{loadtight}
If some bin $v$ has load at least $k+1$ then either the hash graph contains a $B_k$ or the $(k+1)$'st load graph $(V,E)$ of $v$ will satisfy  that $|V|\leq |E|(d-1)-1$ i.e.~be tight.
\end{lemma}
Now if the $(k+1)$'st load graph is tight, the fact that it has height at most $k+1$ implies that it actually contains a tight subgraph of size $O(k)$ as is shown in the following lemma. 
\begin{lemma}\label{smalltight}
Suppose some node $v$ has load at least $k+1$. Then either the hash graph will contain a $B_k$ or a tight subgraph $G'=(E',V')$ with $|E'|=O(k)$.
\end{lemma}
\begin{proof}
If the $(k+1)$'st load graph $G=(V,E)$ of $v$ does not contain a $B_k$ we may by the Lemma~\ref{loadtight} assume that it is tight i.e.~has $|V|\leq |E|(d-1)-1$. Now define $S_0=(\{v\},\emptyset)$ where $v$ is the node of load $(k+1)$ and recursively let $S_i=(V_i,E_i)$ where $E_i=\{e\in E: \exists w\in V_{i-1} \text{ such that  } w\in e\}$ and $V_i=\bigcup_{e\in E_i}e$. Note that since the load graph has height at most $k+1$ we must have $(V_{k+1},E_{k+1})=(V,E)$ so the process stops after at most $k+1$ steps. 

Enumerate the edges of $E$, $e_1,\dots,e_{|E|}$, in any way satisfying that if $e_\ell\in E_i$ and $e_{\ell'}\in E_j\backslash E_i$ for some $i<j$ then $\ell <\ell'$ i.e.~according to (this measure of) distance from $v$. Suppose we construct $(V,E)$ by adding the edges $e_1,\dots,e_{|E|}$ one at a time. Let the graph obtained after the $i$'th edge is added be denoted $(V_i',E_i')$. This process will at any stage give a connected graph thus satisfying $|V_i'|\leq |E_i'|(d-1)+1$ and since $|V|\leq |E|(d-1)-1$ there will exist a minimal $i$ and a minimal $j\geq i$ (possibly with $j=i$) such that $|V_i'|\leq |E_i'|(d-1)$ and $|V_j'|\leq |E_j'|(d-1)-1$. 

If $i=j$ we have that $|e_i\cap V_{i-1}'|\geq 3$ so we can pick three vertices $v_1,v_2,v_3\in e_i\cap V_{i-1}'$. Since  $(V_{i-1}',E_{i-1}')$ is connected and has height at most $k+1$ the smallest connected subgraph $H$ containing $v_1,v_2$ and $v_3$ has itself size $O(k)$. Then $H\cup \{e_i\}$ will be a tight subgraph of size $O(k)$.

When $i<j$ we in a similar way see that when adding $e_i$ we obtain a subgraph $H$ of size $O(k)$ with $|V(H)|= (d-1)|E(H)|$. Next $|e_j\cap V_{j-1}'|\geq 2$ so we can find $v_1,v_2\in e_j\cap V_{j-1}'$. The smallest connected subgraph of $(V_{j-1}',E_{j-1}')$ containing $v_1,v_2$ and $H$ has size $O(k)$ and adding the edge $e_j$ gives a tight subgraph of size $O(k)$.
\end{proof}
\subsection{The Fibonacci trees}

Now suppose that we are in the setting of Theorem~\ref{alwaysleftthm}. We will need to redefine what we mean by the load graph of a bin. It will be silly to use the old definition for the following reason: Consider a node $v$ say in the $i$'th table and suppose we want to know how it got its $k$'th ball. We then consider the corresponding hyperedge $e$ which has a node in each of the $d$ tables. Call these nodes $v_1,\dots,v_d$. Since we use the Always-Go-Left algorithm the bins corresponding to $v_1,\dots,v_{i-1}$ already has load $k$ and thus we reduce the potential size of our witness tree by only asking how they got load $k-1$. We thus define the load graph of a bin as follows.
\begin{figure}
  \centering
  \def\svgwidth{0.8\linewidth}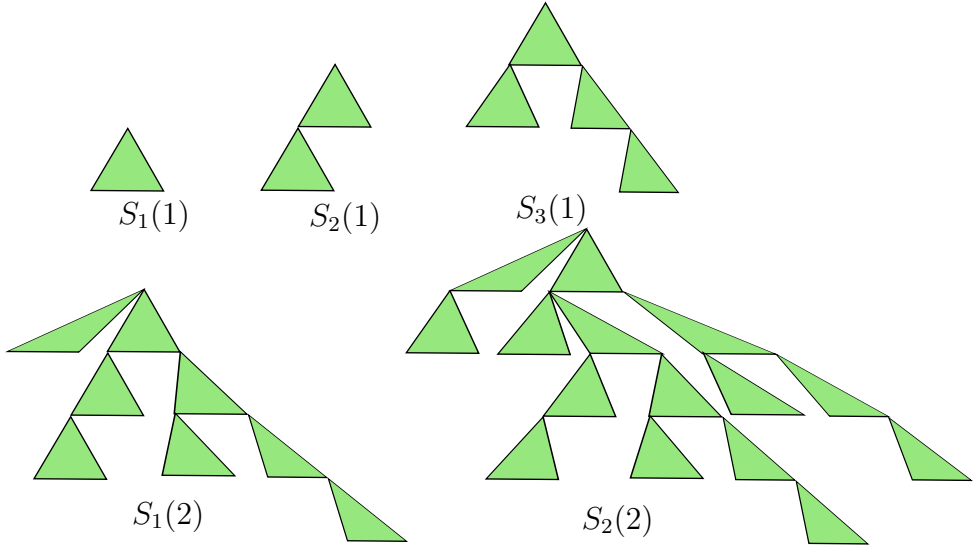
      \caption{The first few $3$-ary Fibonacci trees.}
      \label{fibtrees}
\end{figure}
\begin{definition}
Suppose $v$ is a vertex of the hash graph corresponding to a bin of load at least $k$. We recursively define  $L_v(k)$ the  \textbf{$k$'th load graph of $v$} to be the following $d$-uniform hypergraph.
\begin{itemize}
\item If $k=0$ we let $L_v(k)=(\{v\},\emptyset)$.
\item If $k>0$ and $v\in G_i$ we let $e$ be the edge corresponding to the $k$'th ball landing in $v$. Write $e=\{v_1,\dots,v_d\}$ such that $v_j\in G_j$ for each $j$ (note $v_i=v$). Then $L_v(k)$ is the graph having
\begin{align*}
E(L_v(k))&=\{e\}\cup\bigcup_{j=1}^{i-1}E(L_{v_j}(k))\cup \bigcup_{j=i}^{d}E(L_{v_j}(k-1)), \quad \text{and} \\
V(L_v(k))&=\bigcup_{e\in E(L_v(k))}e
\end{align*}
\end{itemize}
\end{definition}
Note that the $k$'th load graph of a vertex $v\in G_i$ has height at most $d(k-1)+i$ (except of course when $k=0$ in which case the height is zero). 

Next, we will define our witness trees (see Figure~\ref{fibtrees}).
\begin{definition}
For $1\leq i \leq d$ define the $d$-ary Fibonacci tree $S_i(k)$ rooted at a vertex $v$ recursively as follows.
\begin{itemize}
\item When $k=0$ we let $S_i(k)=(\{v\},\emptyset)$.
\item For $k>0$ we let $S_i(k)$ consist of an edge $e=(v_1,\dots,v_d)$ such that $v_j$ is itself a root of an $S_j(k)$ for $1\leq j\leq i-1$ and an $S_j(k-1)$ for $i\leq j\leq d$.
\end{itemize}
\end{definition}

The following result is proved exactly like Lemma~\ref{smalltight}.
\begin{lemma}\label{smalltight2}
Suppose some node $v\in G_i$ has load at least $k+1$. Then either the hash graph contains a tight subgraph of size $O(k)$ or it contains a copy of $S_i(k)$.
\end{lemma}

\section{Bounding the probability of the existence of a large $d$-nomial tree}\label{Tree1app}
Mimicking the methods in~\cite{Dahl1} we will prove the following result\footnote{Some authors say that an event occurs with high probability if the failure probability is $o(1)$. In this terminology Theorem~\ref{nobinom} can be considered a high probability bound on the maximum load.}.
\begin{theorem}\label{nobinom}
There exists a constant $r=O(1)$ such that when hashing $m=O(n)$ balls into $d$ tables of size $n/d$ using $d$ simple tabulation hash functions the probability that the hash graph contains a $d$-nomial tree of size $k=\lceil \frac{\log \log n}{\log d}+r \rceil$ is $O((\log \log n)^{-1})$.
\end{theorem}
In Appendix~\ref{finishingapp} we will see how to deduce Theorem~\ref{mythm}.

When bounding the probability of having a large $d$-nomial tree in the hashgraph we will actually upper bound it by the probability of finding the following $\ell$-\emph{pruned} tree for a fixed $\ell$ (see Figure~\ref{pruned}). 
\begin{definition}\label{prundeddef}
For $k\geq 0$ and $0\leq \ell \leq k$ let the $\ell$-pruned $d$-nomial tree $T_{k,\ell}$ be the tree obtained from $B_k$ by for each vertex $w$ of $B_k$ such that $w$ has less than $(d-1)\ell$ children removing the edges of the induced subtree rooted at $w$ (and the thus created isolated vertices).
\end{definition}
\begin{figure}
\centering
    \includegraphics[width=0.7\textwidth]{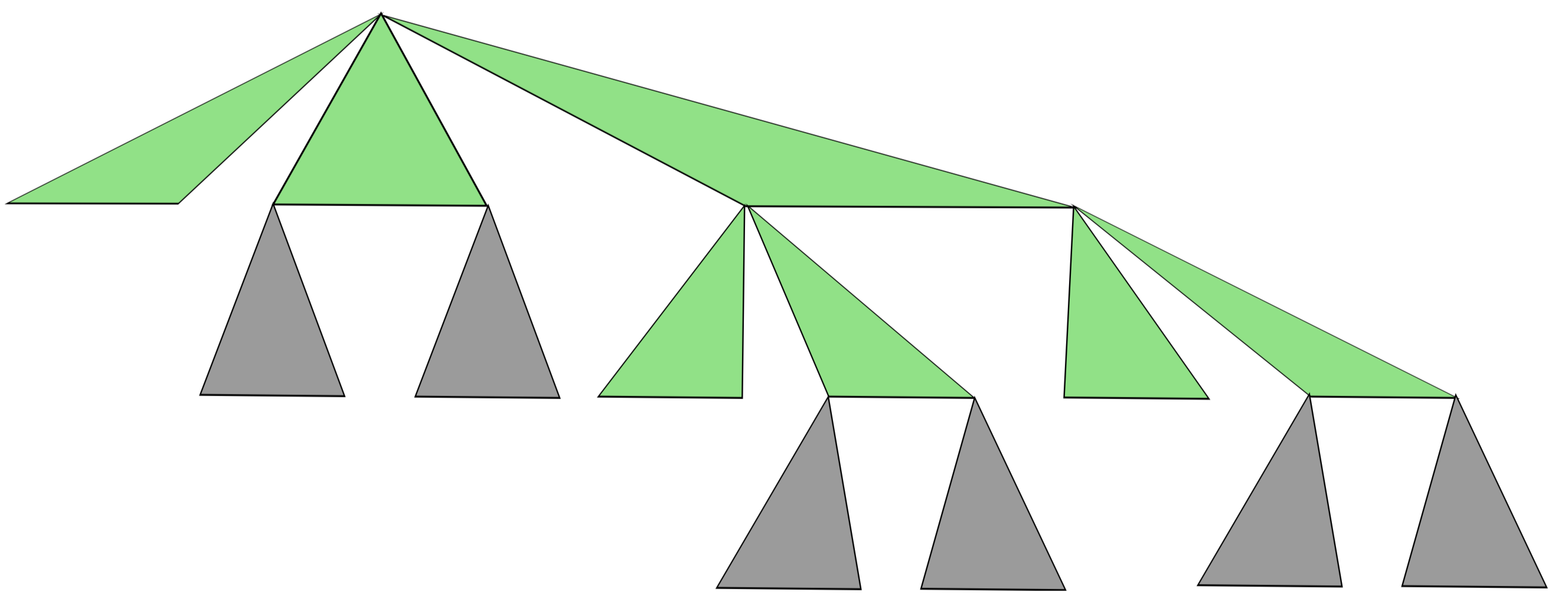}
      \caption{The $2$-pruned 3-nomial tree $T_{2,3}$.}
      \label{pruned}
\end{figure}

Note that each internal node $v\in T_{k,\ell}$ is contained in $\ell$ edges going to children of $v$ that are all leaves. Furthermore, the following results are easily shown by induction starting with the case $k=\ell$.
\begin{lemma}\label{dnomprop} The following holds:
\begin{itemize}
\item $|V(T_{k,\ell})|=((d-1)\ell+1)d^{k-\ell}$ and $|E(T_{k,\ell})|=\ell d^{k-\ell}+\frac{d^{k-\ell}-1}{d-1}$.
\item The number of internal notes in $T_{k,\ell}$ is $d^{k-\ell}$.
\end{itemize}
\end{lemma}

Finally, to prove Theorem~\ref{nobinom} we will need the following two structural lemmas from~\cite{Dahl1}. 
\begin{lemma}[Dahlgaard et al.~\cite{Dahl1}]\label{struclem1}
Let $X\subset U$ with $|X|=m$ and let $s$ be fixed such that $s^c\leq \frac{4}{5}m$. Then the number of $s$-tuples $(x_1,\dots,x_s)\in X^s$ for which there is a $y\in X\backslash \{x_1,\dots ,x_s\}$ such that $h(y)$ is dependent of $h(x_1),\dots,h(x_s)$ is at most
\begin{align*}
s^{O(1)}m^{s-1}.
\end{align*}
\end{lemma}
\begin{lemma}[Dahlgaard et al.~\cite{Dahl1}]\label{struclem2}
Let $X\subset U$ with $|X|=m$ and let $s$ be fixed such that $s^c\leq \frac{4}{5}m$. Let $k\geq \max(s-1,5)$. Then the number of $s$-tuples $(x_1,\dots,x_s)\in X^s$ for which there are $y_1,\dots,y_k\in X\backslash \{x_1,\dots ,x_s\}$  such that each $h(y_i)$ is dependent of $h(x_1),\dots,h(x_s)$ is at most
\begin{align*}
s^{O(1)}m^{s-3/2}.
\end{align*}
\end{lemma}
Now we commence the proof of Theorem~\ref{nobinom}.
\begin{proof}[Proof of Theorem~\ref{nobinom}]
Let $k=\lceil \frac{\log \log n}{\log d}+r \rceil$ for some constant $r$ to be determined later depending only on $c$ and the size of the implicit constant in $m=O(n)$. Suppose that the hash graph contains a $B_k$. Then it also contains a $T_{k,\ell}$ for some $\ell$ that we will fix soon. We will split the analysis into several cases according to the dependencies of the keys hashing to the $T_{k,\ell}$.
\paragraph*{Case 1: The keys hashing to $T_{k,\ell}$ are mutually independent.} Let $s=|E(T_{k,\ell})|$. Note that each of the internal nodes of $T_{k,\ell}$ is contained in exactly $\ell$ edges going to children of $w$ such that these children are all leaves. The number of ways we can choose the keys hashing to $T_{k,\ell}$, including their order, is thus by Lemma~\ref{dnomprop} at most 
\begin{align*}
\frac{m^s}{(\ell!)^{d^{k-\ell}}}.
\end{align*}
The probability that such a choice of keys actually hash to the desired positions is at most $\left(\frac{d^2}{n}\right)^{s-1}$ and by a union bound the probability that the hash graph contains a $T_{k,\ell}$ consisting of independent edges is at most 
\begin{align*}
\frac{m^s}{(\ell!)^{d^{k-\ell}}}\left(\frac{d^2}{n}\right)^{s-1}=\frac{n}{d^2}\left(\frac{d^2m}{n(\ell!)^{\frac{d^{k-\ell}}{s}}}\right)^s\leq \frac{n}{d^2}\left(\frac{d^2m}{n\sqrt[\ell+1]{\ell!} }\right)^s,
\end{align*}
using, in the last step, that $s=\ell d^{k-\ell}+\frac{d^{k-\ell}-1}{d-1}\leq (\ell+1)d^{k-\ell}$. Now if $\ell=O(1)$ is chosen such that $\frac{d^2m}{n\sqrt[\ell+1]{\ell!}}<\frac{1}{2}$, which is possible since $m=O(n)$ and $d=O(1)$, and $r$ is chosen such that $r\geq \ell$ then $s> \ell d^{k-\ell}\geq \ell \log n $ and we get that the probability is at most 
\begin{align*}
\frac{n}{d^2} 2^{-s}< \frac{1}{n^{\ell-1}}\leq\frac{1}{n}
\end{align*}
if $\ell\geq 2$. This suffices and completes case 1.

In the next cases we will bound the probability that the hash graph contains a $T_{k,\ell}$ consisting of \emph{dependent} keys. From such a tree we construct a set $S$ of independent edges as follows: Order the edges of $T_{k,\ell}$ in increasing distance from the root and on each level from left to right\footnote{The meaning of this should be clear by considering Figure~\ref{pruned}. The important thing is that when we have added an edge going to the children of a vertex $v$ we in fact add \emph{all} such edges before continuing the procedure.}. 
Traversing the edges in this order we add an edge to the set $S$ if the corresponding key is independent of all the keys corresponding to edges already in $S$. Stop the process as soon as we meet an edge dependent on the keys in $S$. As we are not in case 1, the process stops before all keys are added to $S$.
\paragraph*{Case 2: All edges incident to the root lie in $S$.} Let $s=|S|$ be fixed. Let us first count the number of ways to choose the elements of $S$ accounting for symmetries in the corresponding subset of $T_{k,\ell}$. First of all note that $s=O( \log m)$ so we can apply Lemma~\ref{struclem1} and conclude that the set $S$, including the order, can be chosen in at most $s^{O(1)}m^{s-1}$. Despite saving a factor of $m$ a direct union bound will not suffice but we are close and we have not yet taken  advantage of the symmetries of the subset of $T_{k,\ell}$. 

Now by the way we traverse the edges when constructing $S$ there can be at most one internal node $v$ of $S$ contained in less than $\ell$ edges going to children of $v$ that are all leaves. If $v_1,\dots,v_h$ denote the internal vertices of $S$ and $w_i$ denotes the number of edges containing $v_i$ and going to children of $v_i$ that are all leaves we therefore have that $w_i<\ell$ for at most one $i$. 

With this definition the number of ways to choose $S$ is at most 
\begin{align*}
s^{O(1)}m^{s-1}\prod_{i=1}^h\frac{1}{w_i!}\leq s^{O(1)}m^{s-1}\prod_{i=1}^h\left(\frac{e}{w_i}\right)^{w_i}.
\end{align*}
Now $f:x\mapsto x \log (e/x)$ is concave ($f''(x)=-1/x<0$) so by Jensen's inequality we obtain
\begin{align*}
\prod_{i=1}^h\left(\frac{e}{w_i}\right)^{w_i}=\exp\left(\sum_{i=1}^h w_i \log (e/w_i)\right)\leq \exp\left(w \log \left( \frac{eh}{w}\right)\right)=\left(\frac{eh}{w}\right)^w
\end{align*}
where $w=\sum_{i=1}^hw_i$. 

We may assume that $v_1$ is the root and since the keys adjacent to $v$ are all in $S$ we have that $w_1\geq \ell$ even if $h=1$. We thus get that 
\begin{align*}
w=w_1+\cdots w_h\geq \begin{cases}
\ell, &h=1, \\
(h-1)\ell, &h\geq 2
\end{cases}
\end{align*}
Hence, in any case we obtain that $h/w\leq 2/\ell$. Secondly, for $h\geq 2$ we have that $w\geq s-h\geq s-\frac{w+\ell}{\ell}$ so $w\geq (s-1)\frac{\ell}{\ell+1}$. When $h=1$ we have the even stronger bound $w\geq s$. We thus obtain, assuming $\ell >2e$, that
\begin{align*}
\left(\frac{eh}{w}\right)^w\leq \left( \frac{2e}{\ell}\right)^{(s-1)\frac{\ell}{\ell+1}}.
\end{align*}
We now assume that $\ell$ is so large that $\left( \frac{2e}{\ell}\right)^{\frac{\ell}{\ell+1}}<\frac{n}{2d^2m}$. Then the number of ways to choose $S$ is at most
\begin{align*}
s^{O(1)}\left( \frac{n}{2d^2} \right)^{s-1}.
\end{align*}
Like in case 1 the probability that one of these choices of keys actually hash to $S$ is at most $\left( \frac{d^2}{n}\right)^{s-1}$ and so by a union bound we get that the probability of the event in case 2, for fixed $s$, is bounded by
\begin{align*}
s^{O(1)}2^{1-s}.
\end{align*}
A union bound over all $s> \frac{\log \log n}{\log d}$ gives that the probability of the event in case 2 is at most
\begin{align*}
&\sum_{s>\frac{\log \log n}{\log d}} s^{O(1)}2^{1-s}\leq 2^{-\frac{\log \log n}{\log d}+2}\sum_{k\geq 1}2^{-k}\left(k+\frac{\log \log n}{\log d}-1\right)^{O(1)} \\
= & \frac{4}{(\log n)^{1/\log d}} \left(\frac{\log \log n}{\log d} \right)^{O(1)}=  \frac{(\log \log n)^{O(1)}}{(\log n)^{1/\log d}}
\end{align*}
which suffices.

We may now assume that not all of the edges incident to the root are independent and we will let $S'$ be a largest set of independent edges incident to the root. We divide the proof into two cases.
\paragraph*{Case 3: Not all but at least $\frac{\log \log n}{2\log d}$ edges incident to the root lie in $S'$.}
The proof in case 3 is almost similar to the proof in case 2 but much simpler. The reason we need it is that it allows us to assume that we have a lot of edges dependent on the edges in $S'$ adjacent to the root and thus use Lemma~\ref{struclem2}. \\
Let $s'=|S'|$ be fixed. The number of ways we can choose the keys in $S'$ (including their order) is by Lemma~\ref{struclem1} bounded by $\frac{s'^{O(1)}m^{s'-1}}{s'!}$ so the probability of finding such a set is at most
\begin{align*}
\frac{s'^{O(1)}m^{s'-1}}{s'!}d\left(\frac{d}{n} \right)^{s'-1} \leq d s'^{O(1)} \left(\frac{med}{ns'} \right)^{s'-1}= s'^{O(1)} O((\log n)^{-1}),
\end{align*}
using in the last step that $s'=\Omega(\log \log n)$.
A union bound over all $s'\leq \frac{\log \log n}{\log d}+r=O( \log \log n)$ gives the desired.
\paragraph*{Case 4: Less than $\frac{\log \log n}{2\log d}$ edges incident to the root lie in $S'$.}
By Lemma~\ref{struclem2} the number of ways to choose the keys in $S'$ is at most $s'^{O(1)}m^{s'-3/2}$. Thus the probability that such a set $S'$ occurs is (not even accounting for the symmetries)  at most
\begin{align*}
s'^{O(1)}m^{s'-3/2}d\left( \frac{d}{n} \right)^{s-1}\leq ds'^{O(1)}\left(\frac{dm}{n} \right)^{s'-3/2}n^{-1/2}= (\log n)^{O(1)}n^{-1/2}.
\end{align*}
Summing over all $s'$ gives the desired result and the proof is complete.
\end{proof}

\section{Bounding the probability of the existence of a Fibonacci tree}\label{Tree2app}
We will prove the following result. In Appendix~\ref{finishingapp} we will see how to deduce Theorem~\ref{alwaysleftthm}.
\begin{theorem}\label{nofibtree}
There exists a constant $r=O(1)$ such that when hashing $m=O(n)$ balls into $d$ tables of size $g=n/d$ using $d$ simple tabulation hash functions the probability that the hash graph contains an $S_i(k)$ of size at least $k=\lceil \frac{\log \log n}{d\log \varphi_d}+r \rceil$ is $O((\log \log n)^{-1})$.
\end{theorem}

The proof of Theorem~\ref{nofibtree} is very similar to the proof of Theorem~\ref{nobinom}. First of all let us define the $\ell$-pruned version of $S_i(k)$. The definition is analogous to definition~\ref{prundeddef} 
\begin{definition}
For $k,\ell \geq 0$ let $P_{i,\ell}(k)$ be the tree obtained from $S_i(k)$ by for each vertex $w$ with less than $(d-1)\ell$ children removing the edges of the induced subtree rooted at $w$ (and the thus created isolated vertices).
\end{definition}
Like in the proof of Theorem~\ref{nobinom} we will thus need to know the number of edges of $P_{i,\ell}(k)$ as well as the number of internal vertices $w$ of $P_{i,\ell}(k)$ such that $w$ is contained in at least $\ell$ edges going to children of $v$ that are all leaves. In that direction we have the following result.
\begin{lemma}\label{usefulfacts}
The following holds
\begin{enumerate}
\item The number of edges of $P_{i,\ell}(\ell)$ is exactly $\ell 2^{i-1}$. Also, for $k>\ell$ we have that 
\begin{align*}
|E(P_{i,\ell}(k))|=1+\sum_{j=1}^{i-1}|E(P_{j,\ell}(k))|+\sum_{j=i}^{d}|E(P_{j,\ell}(k-1))|.
\end{align*}
In particular for $k\geq \ell$
\begin{align*}
F_d(d(k-\ell)+i+1)\leq\frac{ |E(P_{i,\ell}(k))|}{\ell}\leq F_d(d(k-\ell)+i+2).
\end{align*}
\item The number $g_{i,\ell}(k)$ of vertices $w$ of $P_{i,\ell}(k)$ that are contained in at least $\ell$ edges going to children of $w$ that are leaves is exactly $F_d(d(k-\ell)+i)$.
\end{enumerate}
\end{lemma}
\begin{proof}
Let us prove 1. first. Clearly $|E(P_{i,\ell}(k))|=0$ when $k<\ell$ and $|E(P_{1,\ell}(\ell))|=\ell$. It is also  easy to check the recursion $|E(P_{i,\ell}(\ell))|=\ell+\sum_{j=1}^{i-1}|E(P_{j,\ell}(\ell))|$ which implies that $|E(P_{i,\ell}(\ell))|=\ell2^{i-1}$. The last equality follows from the fact that for $k>\ell$ we have that $P_{i,\ell}(k)$ consists of one edge and a copy of $P_{j,\ell}(k)$ for each $j<i$ together with a copy of $P_{j,\ell}(k-1)$ for  each $j>i$ and these are all being $\ell$-pruned in the process of $\ell$-pruning $S_i(k)$

Finally let's prove the estimate on $|E(P_{i,\ell}(k))|$. The lower bound clearly holds when $k=\ell$ (here we have equality) and for $k>\ell$ we inductively have that 
\begin{align*}
|E(P_{i,\ell}(k))|>\sum_{j=1}^{i-1}|E(P_{j,\ell}(k))|+\sum_{j=i}^{d}|E(P_{j,\ell}(k-1))|\geq \ell F_d(d(k-\ell)+i+1).
\end{align*}
Now for the upper bound. Let $\alpha_{i,\ell}(k)=|E(P_{i,\ell}(k))|+\frac{1}{d-1}$. Then for $k>\ell$ we have that 
\begin{align*}
\alpha_{i,\ell}(k)=\sum_{j=1}^{i-1}\alpha_{j,\ell}(k)+\sum_{j=i}^{d}\alpha_{j,\ell}(k-1).
\end{align*}
It is trivial to check that $\alpha_{i,\ell}(\ell)\leq \ell F_d(i+2)$ and this combined with the recursion gives that $\alpha_{i,\ell}(k)\leq \ell F_d(d(k-\ell)+i+2)$ for any $ k\geq \ell$ so we get the stated inequality.

Now for the second statement. When $k<\ell$ the number of such vertices is zero so the result is trivial. Also, when $k=\ell$ and $i=1$ there is exactly $1=F_d(1)$ such vertex. Finally, consider the tree $P_{i,\ell}(k)$ for $k\geq \ell$ and $(k,i)\neq (\ell,1)$. The root $v$ is  contained in $k\geq \ell$ edges so these are not pruned. 

Now $S_i(k)$ consist of an edge $e=(v_1,\dots,v_d)$ such that $v_j$ is a root of an $S_j(k)$ for $j<i$ and an $S_j(k-1)$ for $j\geq i$. 

Suppose first that $k=\ell$. In the process of pruning $S_i(\ell)$ we prune $S_1(\ell),\dots,S_{i-1}(\ell)$. Hence, 
\begin{align*}
g_{i,\ell}(\ell)= \sum_{j=1}^{i-1}g_{j,\ell}(\ell)= \sum_{j=1}^{i-1}F_d(j)=F_d(i).
\end{align*}
A similar argument works when $k>\ell$. In this case we prune the subtrees $S_j(k)$ for $j<i$ and the subtrees $S_j(k-1)$ for $j\geq i$ so we get
\begin{align*}
g_{i,\ell}(\ell)= \sum_{j=1}^{i-1}g_{j,\ell}(k)+\sum_{j=i}^{d}g_{j,\ell}(k-1)= F_d(d(k-\ell)+i)
\end{align*}
and we are done. 
\end{proof}

Now we are ready to prove Theorem~\ref{nofibtree}. 
\begin{proof}[Proof of Theorem~\ref{nofibtree}]
Like in the proof of Theorem~\ref{nobinom} we will split the proof in four cases.
\paragraph*{Case 1: The keys hashing to $P_{i,\ell}(k)$ are mutually independent.} Let $s=|E(P_{i,\ell}(k))|$. The number of ways to choose the keys hashing to $P_{i,\ell}(k)$ (including their positions) is like in the proof of Theorem~\ref{nobinom} at most 
\begin{align*}
\frac{m^s}{(\ell!)^{g_{i,\ell}(k)}}= \frac{m^s}{(\ell!)^{F_d(d(k-\ell)+i)}}
\end{align*}
where we used Lemma~\ref{usefulfacts}. Hence, by a union bound the probability of having an $P_{i,\ell}(k)$ consisting of independent keys is at most
\begin{align*}
\frac{m^s}{(\ell!)^{F_d(d(k-\ell)+i)}}\left(\frac{d^2}{n}\right)^{s-1}=\frac{n}{d^2}\left( \frac{d^2m}{n(\ell!)^{\frac{F_d(d(k-\ell)+i)}{s}}}\right)^s
\end{align*}
where $s=|E(P_{i,\ell}(k))|$. But by the inequality in Lemma~\ref{usefulfacts} we know that
\begin{align*}
\frac{F_d(d(k-\ell)+i)}{s}\geq \frac{F_d(d(k-\ell)+i)}{\ell F_d(d(k-\ell)+i+2)}\geq \frac{1}{4\ell}.
\end{align*}
Hence, choosing $\ell$ sufficiently large we get that the probability above is at most $\frac{n}{d^2}2^{-s}$. Now $\varphi_d$ is the rate of growth of $F_i(d)$ and we can find a constant $c_d$ such that $F_d(i)\geq c_d\varphi^{i}$ for all $i\in \mathbb{N}$. Thus
\begin{align*}
2^s\geq 2^{F_d(d(k-\ell))}\geq 2^{c_d\varphi^{d(k-\ell)}}.
\end{align*}
It follows that if $k\geq \frac{1+\log \log n-\log c_d}{d\log \varphi_d}+\ell$ then $2^{-s}\leq n^{-2}$ and since $r=\frac{1-\log c_d}{d\log \varphi_d}+\ell=O(1)$ we are done.
\paragraph*{Case 2: All edges incident to the root are independent.}
We proceed as in the proof of Theorem~\ref{nobinom} by constructing a set $S$ of independent keys in the following way. We order the edges of $P_{i,\ell}(k)$ according to increasing distance to the root and on each level from left to right. We then traverse the edges in this order adding a key to $S$ if it is independent on the keys already in $S$. We stop the process as soon as we meet a dependent key. Like in the proof of Theorem~\ref{nobinom} we let $v_1,\dots,v_h$ denote the internal nodes of $S$ and for $1\leq i \leq h$ we let $w_i$ denote the number of edges containing $v_i$ and going to children of $v_i$ that are all leaves. Then using Lemma~\ref{struclem1} we conclude that the number of ways to choose the keys (including their position) is at most
\begin{align*}
s^{O(1)}m^{s-1}\prod_{i=1}^h \frac{1}{w_i!}\leq s^{O(1)} m^{s-1}\left(\frac{eh}{w}\right)^w
\end{align*}
where $s=|S|$. When bounding $\frac{h}{w}$ we cannot proceed exactly as in the proof of Theorem~\ref{nobinom} because there might be many internal nodes (not just one) of $S$ that are not the starting node of at least $\ell$ edges going to lower level leaves. However, we only need to change the argument slightly and by doing so it will actually also work for case 2 in the proof of Theorem~\ref{nobinom}. 

$S$ is constructed by first adding all the edges adjacent to the root to $S$ and then repeatedly adding groups of at least $\ell$ edges going to children of a given node $v$. Finally we add a group of edges, which might have size $<\ell$, to a leaf (making it an internal node). Let these steps be enumerated $1,\dots,t$ for some $t$.

Let $h_j$ denote the number of internal nodes after the $j$'th step. Similarly, after the $j$'th step, denote by $\alpha_j$ the number of edges $e$ containing a vertex $v$ and going to children of $v$ such that all the children of $v$ lying in $e$ are leaves.

Clearly $\frac{h_1}{\alpha_1}=\frac{1}{k}\leq \frac{1}{\ell-1}$. Also, for $j<t$ we have that $h_{j}\leq h_{j-1}+1$ and $\alpha_j\geq \alpha_{j-1}+\ell-1$. Hence, if $\frac{h_{j-1}}{\alpha_{j-1}}\leq \frac{1}{\ell-1}$ we must have that 
\begin{align*}
\frac{h_{j}}{\alpha_{j}}\leq \frac{h_{j-1}+1}{\alpha_{j-1}+\ell-1}\leq \frac{1}{\ell-1}
\end{align*} 
so this inequality is preserved. Finally, when adding the $t$'th group (which might have size smaller than $\ell$) we don't change this inequality by much. Indeed,
\begin{align*}
\frac{h}{w}=\frac{h_t}{\alpha_t}\leq \frac{h_{t-1}+1}{\alpha_{t-1}}\leq \frac{1}{\ell-1}+\frac{1}{\alpha_{t-1}}\leq \frac{1}{\ell-1}+\frac{1}{k} \leq \frac{2}{\ell}
\end{align*}
if $\ell=O(1)$ and $k$ are sufficiently large. 

Defining $s_j$ to be the total number of edges after the $j$'th group is inserted we in a similar way see that for $j<t$
\begin{align*}
\alpha_j\geq \frac{\ell-1}{\ell}s_j
\end{align*}
and so $w=\alpha_t\geq \frac{\ell-1}{\ell}(s-1)$. Hence, if $\ell>2e$, the number of ways to choose the independent keys, including their position, is at most 
\begin{align*}
s^{O(1)}m^{s-1}\left(\frac{2e}{\ell} \right)^{\frac{\ell-1}{\ell}(s-1)}
\end{align*}
and from here on the proof is identical to the proof of Theorem~\ref{nobinom}.

\end{proof}

\section{Completing the proofs}\label{finishingapp}
In this appendix we wrap up the proofs of Theorem~\ref{mythm} and Theorem~\ref{alwaysleftthm}. Combining Lemma~\ref{smalltight}, Theorem~\ref{tightthm} and Theorem~\ref{nobinom} we see that there is a constant $r>0$ such that the probability that the maximum load $L$ is at least $\frac{\log \log n}{\log d}+r+1$ is $O((\log \log n)^{-1})$. To see that this suffices we first recall the high probability bound by Dahlgaard et al.~\cite{Dahl1}.

\begin{theorem}[Dahlgaard et al.~\cite {Dahl1}]\label{dahlthm}
Let $h_1$ and $h_2$ be two independent random simple tabulation hash functions. If
$m = O(n)$ balls are placed in two tables each consisting of $n/2$ bins sequentially using the two-choice paradigm
with $h_1$ and $h_2$, then for any constant $\gamma > 0$, the maximum load of any bin is $O(\log \log n)$ with
probability $1-O(n^{-\gamma})$.
\end{theorem}

Using this result we in fact get that even with $d$ choices the maximum load is $O(\log \log n)$ whp. Indeed, if there is a way to insert the $m$ balls into $d$ groups $G_1,\dots,G_d$ using $h_1,\dots,h_d$ respecting the $d$-choice paradigm and obtaining a maximum load of $L$, it is easy to check that if we insert the same balls into $G_1$ and $G_2$ restricting our choices to $h_1$ and $h_2$  and using the two choice paradigm we can obtain a maximum load of at least $L$. Since $m=O(n/d)$ (as $d$ is constant) Theorem~\ref{dahlthm} applies.

 Thus, there is an $\alpha>0$ such that the probability that the maximum load is at least $\alpha \log \log n$ is at most $n^{-1}$. Putting $k=\frac{\log \log n}{\log d}+r$ we obtain that
\begin{align*}
\E L=\sum_{i=1}^{k} \P(L\geq i)+\sum_{i=k+1}^{\alpha \log \log n} \P(L\geq i)+\sum_{i=\alpha \log \log n+1}^{m} \P(L\geq i)=\frac{\log \log n}{\log d}+O(1),
\end{align*}
which completes the proof of Theorem~\ref{mythm}. A similar argument completes the proof of Theorem~\ref{alwaysleftthm}.

\section{Open problems}\label{openapp}
Several problems concerning the use of simple tabulation in the $d$-choice paradigm remains open. We mention a few here:

\textbf{High probability bounds when $d=\omega(1)$:} The result by Dahlgaard et al.~\cite{Dahl1} implies that when $d=O(1)$ the maximum load is $O(\log \log n)$ whp. What can be said for $d=\omega(1)$? Is the maximum load $O\left(\frac{\log \log n}{\log d} \right)$ whp even when $d=\omega(1)$? In particular, if $d=(\log n)^{\varepsilon}$ for some $\varepsilon>0$ is the maximum load constant? A similar question can be asked for the Always-Go-Left algorithm. 

\textbf{The expected maximum load when $d=\omega(1)$:} Using the same techniques as us but exercising more care one can show that even if $d=\omega(1)$ is allowed to grow \emph{very} slowly the expected maximum load is at most $(1+o(1))\frac{\log \log n}{\log d}$ whp and similarly $(1+o(1))\frac{\log \log n}{d\log \varphi_d}$ for the Always-Go-Left algorithm (we provide no details). Can we obtain a more complete picture? The current techniques bounds the probability of certain \emph{combinatorial} structures in the hash graph that are consequences of the existence of a bin of large load. By this approach they actually yield that regardless of the order of the insertion of the balls the probabilistic bounds remain valid. For large $d$ this seems to be allowing too much adversarial power so other techniques might be needed.

\textbf{The heavily loaded case:} In our analysis we assumed that $m=O(n)$ but what happens for $m\gg n$? Berenbrink et al.~\cite{Ber} has shown that with fully random hashing the maximum load differs from the expected average by at most $\frac{\log \log n}{\log d}+O(1)$ whp. Even for $d=2$ we don't have a similar result with simple tabulation.

\section{The independence of simple tabulation}\label{Indapp}
We will here provide the proofs of Lemma~\ref{depchar} and Lemma~\ref{deplemma} both for completeness and to fairly portray the full length of the new proof of Theorem~\ref{cuckoothm}.
\begin{proof}[Proof of Lemma~\ref{depchar}]
One direction is easy. If $I$ is as described in the lemma $\bigoplus_{i\in I} h(x_i)=0$ as $h(\alpha)$ appears an even number of times in the sum for each position character $\alpha$ and the addition is in a $\Z_2$-vector space. In particular the keys $(x_i)_{i\in I}$ are dependent.

The converse will follow from a translation to linear algebra. Note first that any set $S$ of position characters can be naturally identified with a vector in $\Z_2^{[c] \times \Sigma}$. Indeed, we have a natural bijection $\varphi:\mathcal{P}([c]\times \Sigma)\to \Z_2^{[c] \times \Sigma}$ given by $\varphi:S \mapsto v_S$ where
\begin{align*}
v_S(j,a)=
\begin{cases}
1, & (j,a)\in S \\
0, & (j,a)\notin S
\end{cases}
\end{align*}
Choosing a random simple tabulation hash function is equivalent to uniformly at random picking a linear map $\tilde h:  \Z_2^{[c] \times \Sigma} \to \Z_2^r$ (the identification being $h=\tilde h \circ \varphi$). The assumption on the keys $x_1,\dots,x_k$ is equivalent to saying that the vectors $\varphi(x_1),\dots,\varphi(x_k)$ are linearly independent vectors over $\Z_2$. In particular the hash values $h(x_i)=\tilde h(\varphi(x_i))$ are independent and uniform in $\Z_2^r$.
\end{proof}

\begin{proof}[Proof of Lemma~\ref{deplemma}]
We proceed as in~\cite{Dahl2} and apply induction on $c$. Suppose first that $c=1$. First of all the number of partitions of a set of size $2t$ into $t$ pairs is exactly $(2t-1)!!$. Now, the identity $x_1\oplus \cdots \oplus x_{2t}=\emptyset$ gives that in the sequence $(x_1,\dots,x_{2t}) $ each element appears an even number of times and thus there is a partition of $\{1,\dots,2t\}$ into $t$-pairs $(i_1,j_1),\dots,(i_t,j_t)$ such that $x_{i_\ell}=x_{j_\ell}$ for $1\leq \ell \leq  t$. Now given such a partition the number of ways to choose the $x_i$'s is at most
\begin{align*}
\prod_{\ell=1}^t|A_{i_{\ell}}\cap A_{j_{\ell}}|\leq \prod_{\ell=1}^t \min (|A_{i_{\ell}}|,|A_{j_{\ell}}|)\leq\prod_{\ell=1}^t\sqrt{|A_{i_\ell}|}\sqrt{|A_{j_\ell}|}= \prod_{i=1}^{2t}\sqrt{|A_i|}.
\end{align*}
Summing over all $(2t-1)!!$ partitions gives the desired upper bound. 

Now suppose $c\geq 2$ and that the result holds for smaller $c$. We write $x=(x[0],\dots,x[c-1])$ for $x\in U$. For $a\in \Sigma$ we define $A_i[a]=\{x\in A_i:x[0]=a\}$. Then for a fixed partition of $\{1,\dots,2t\}$ into pairs $(i_1,j_1),\dots,(i_t,j_t)$ and for fixed choices of $a_1,\dots,a_t$ the induction hypothesis gives that the number of $2t$-tuples $(x_1,\dots,x_{2t})\in A_1\times \cdots \times A_{2t}$ with $x_1\oplus \cdots \oplus x_{2t}=\emptyset$ such that $x_{i_\ell}[0]=x_{j_\ell}[0]=a_\ell$ is at most 
\begin{align*}
((2t-1)!!)^{c-1}\prod_{\ell=1}^{t}\sqrt{|A_{i_\ell}[a_\ell]|}\sqrt{|A_{j_\ell}[a_\ell]|}.
\end{align*}
We sum this over all partitions and all choices of $a_1,\dots,a_t$ to get a total upper bound on the number of $2t$-tuples $(x_1,\dots,x_{2t})\in A_1 \times \cdots \times A_{2t}$ such that $x_1\oplus \cdots \oplus x_{2t}=\emptyset$ of
\begin{align*}
&((2t-1)!!)^c\sum_{a_1,\dots,a_{t}}\prod_{\ell=1}^{t}\sqrt{|A_{i_\ell}[a_\ell]|}\sqrt{|A_{j_\ell}[a_\ell]|}=((2t-1)!!)^c\prod_{\ell=1}^t\left(\sum_a  \sqrt{|A_{i_\ell}[a]|}\sqrt{|A_{j_\ell}[a]|}\right) \\
\leq& ((2t-1)!!)^c\prod_{\ell=1}^t \sqrt{\sum_a |A_{i_\ell}[a]|}\sqrt{\sum_a |A_{j_\ell}[a]|} =((2t-1)!!)^c \prod_{\ell=1}^t \sqrt{|A_{i_\ell}|}\sqrt{|A_{j_\ell}|}\\
=&((2t-1)!!)^c\prod_{i=1}^{2t}\sqrt{|A_i|},
\end{align*}
where we used Cauchy-Schwartz's inequality in the second step. This completes the induction.
\end{proof}

\end{document}